\documentclass[12pt]{article}
\usepackage{graphicx}
\usepackage{amsfonts}
\usepackage{amsmath}
\usepackage{amssymb}
\usepackage{url}
\usepackage{tikz}
\usepackage{pgfplots}
\pgfplotsset{compat=1.17}
\usepackage{dsfont} 
\usepackage{stmaryrd}

\usepackage{indentfirst}
\usepackage{enumerate}
\usepackage[normalem]{ulem}
\usepackage{mathrsfs}
\usepackage{multirow}
\usepackage{diagbox}
\usepackage{bbold}
\usepackage{stackengine}
\usepackage{caption}
\usepackage{subcaption}
\usepackage[colorlinks=true,citecolor=blue]{hyperref}
\usepackage{amsthm}
\usepackage{color}

\definecolor{DarkGreen}{rgb}{0.2,0.6,0.2}

\definecolor{purple}{rgb}{0.6,0.3,0.8}

\usepackage{natbib} 
\usepackage{comment}
\allowdisplaybreaks

\newcommand{\cX}{\mathcal{X}}

\newcommand{\p}{\mathbb{P}}

\renewcommand{\tilde}{\widetilde}

\DeclareMathOperator*{\argmin}{\arg\min}

\newcommand{\dsquare}{\mathop{  \square} \displaylimits}

\newcommand{\esssup}{\mathrm{ess\mbox{-}sup}}
\newcommand{\essinf}{\mathrm{ess\mbox{-}inf}}
\renewcommand{\cdots}{\dots}
\theoremstyle{plain}
\newtheorem{theorem}{Theorem}
\newtheorem{corollary}{Corollary}
\newtheorem{lemma}{Lemma}
\newtheorem{proposition}{Proposition}
\theoremstyle{definition}

\newtheorem{example}{Example}
\newtheorem{assumption}{Assumption}
\newtheorem{assumptionalt}{Assumption}[assumption]
\newenvironment{assumptionp}[1]{
  \renewcommand\theassumptionalt{#1}
  \assumptionalt
}{\endassumptionalt}

\theoremstyle{remark}
\newtheorem{remark}{Remark}

\usepackage{setspace}

\usepackage{footmisc}
\title{Optimal Allocations with Distortion Risk Measures and Mixed Risk Attitudes}

\author{
Mario Ghossoub\thanks%
  {Department of Statistics and Actuarial Science,
  University of Waterloo,
  Waterloo, Ontario, Canada.
  E-mail: \href{mailto:mario.ghossoub@uwaterloo.ca}{
mario.ghossoub@uwaterloo.ca}.}
  \and Qinghua Ren\thanks%
  {Department of Statistics and Actuarial Science,
  University of Waterloo,
  Waterloo, Ontario, Canada.
  E-mail: \href{mailto:qinghua.ren@uwaterloo.ca}{qinghua.ren@uwaterloo.ca}.}
  \and Ruodu Wang\thanks%
  {Department of Statistics and Actuarial Science,
  University of Waterloo,
  Waterloo, Ontario, Canada.
  E-mail: \href{mailto:wang@uwaterloo.ca}{wang@uwaterloo.ca}.}
  }

\begin{document}
\maketitle

\begin{abstract}
We study Pareto-optimal risk sharing in economies with heterogeneous attitudes toward risk, where agents' preferences are modeled by distortion risk measures. Building on comonotonic and counter-monotonic improvement results, we show that agents with similar attitudes optimally share risks comonotonically (risk-averse) or counter-monotonically (risk-seeking). We show how the general $n$-agent problem can be reduced to a two-agent formulation between representative risk-averse and risk-seeking agents, characterized by the infimal convolution of their distortion risk measures. Within this two-agent framework, we establish necessary and sufficient conditions for the existence of optimal allocations, and we identify when the infimal convolution yields an unbounded value. When existence fails, we analyze the problem under nonnegative allocation constraints, and we characterize optima explicitly, under piecewise-linear distortion functions and Bernoulli-type risks. Our findings suggest that the optimal allocation structure is governed by the relative strength of risk aversion versus risk seeking behavior, as intuition would suggest.
\end{abstract}


\section{Introduction}
In classical risk-sharing problems with risk‐averse agents, Pareto‐optimal allocations are comonotonic under mild conditions, as a result of the classical comonotonic improvement theorem of  \cite{landsberger1994co}. On the other hand, when only risk‐seeking agents are involved, inter-agent gambling arises, thereby leading to counter‐monotonic allocations, as shown in \cite{lauzier2023pairwise,lauzier2024negatively}. 
This contrast naturally leads to the question of how heterogeneous preferences shape optimal allocations in markets that comprise both risk‐averse and risk‐seeking participants, a more realistic assumption in practice. A simple conjecture is that the optimal allocation structure is determined by the relative strength of risk aversion versus risk seeking behavior.

The exchange of risk between agents with mixed risk attitudes remains theoretically underdeveloped. In the literature, risk sharing has been studied primarily in markets with homogeneous preference classes, e.g., all risk-averse, all risk-seeking, or all quantile-type agents. For risk-averse agents, optimal allocations have been characterized in various frameworks, from the seminal work by \cite{borch1962} within the context of risk-averse Expected-Utility (EU) agents, to the case of convex risk measures in \cite{barrieu2005inf} and \cite{jouini2008optimal}, all showing comonotonicity of optimal allocations. The results can also be extended to non-monotone risk measures (\cite{acciaio2007optimal} and \cite{filipovic2008optimal}) and  quasi-convex risk measures (\cite{mastrogiacomo2015pareto}). For more developments on efficient risk sharing  with monetary and convex risk measures or concave utility functionals, see \cite{heath2004pareto}, \cite{tsanakas2009split}, and \cite{ghossoub2024efficiency}, for instance. Although risk aversion serves as a foundational and tractable assumption in theory, empirical observations frequently reveal risk-seeking attitudes in practice. Building upon the counter-monotonic improvement theorem of \cite{lauzier2024negatively}, the recent work of \cite{ghossoub2024counterb,ghossoub2024countera} provided a systematic study of risk sharing among risk-seeking agents with distortion risk measures. They showed that counter-monotonic allocations are optimal in this setting, and they provided some insight into computing the inf-convolution (infimal convolution) of concave risk measures. Beyond convexity or concavity, risk-sharing problems have also been examined in quantile-based and more general distortion risk-metrics frameworks; see \cite{embrechts2018quantile}, \cite{weber2018solvency} and \cite{lauzier2024risk}.

The relationship between heterogeneous risk attitudes and the existence of Pareto optima or competitive equilibria remains an open and challenging question. One should remark that the presence of non-convex preferences can cause technical problems since the existence of optima can fail in this setting. The market that combines risk-averse agents and risk-seeking agents is of particular interest; see \cite{araujo2017optimal,araujo2018general},  
\cite{herings2022competitive}, and \cite{beissner2023optimal}.  In a finite-state exchange economy, \cite{araujo2017optimal} proved the existence of individually rational Pareto optima even in the presence of non-convex preferences, and they suggested that risk-averse agents exhibit comonotonic sharing, while risk-seeking agents bet on the allocation of risk. Within the framework of distortion risk measures, we obtain analogous results. Moreover, rather than providing only qualitative statements, we deliver quantitative characterizations via the inf-convolution. \cite{beissner2024no} studied a pure-exchange economy without aggregate uncertainty, considering two agents who maximize their rank-dependent utility (RDU), with concave utility functions and arbitrary (possibly nonconvex) probability distortion functions. They derived a closed-form characterization of Pareto-optimal allocations in this general setting. Our paper also touches on a related risk-sharing problem with a constant aggregate. Although distortion risk measures can be viewed as a subclass of RDU with linear utility, the results of \cite{beissner2024no} cannot be directly extended to our framework, as they require strict concavity of the utility functions.

To overcome the challenges arising from mixed risk attitudes in the market, we study a finite-agent economy with heterogeneous risk attitudes, where the agents' preferences are represented by distortion risk measures. For such preferences, risk attitudes are characterized by their probability distortion functions. We analyze Pareto-optimal allocations through the tool of inf-convolution. Our first step is to reduce the multi-agent problem into a $2$-agent problem, where one agent acts as the representative agent of the risk-averse participants, and the other acts as the representative agent of the risk-seeking participants (Theorem \ref{proposition_1}). Zooming in on the reduced $2$-agent problem, we provide necessary and sufficient conditions for the existence of optimal allocations. Since unbounded optimal allocations may not exist for risk-seeking agents due to excessive gambling, we consider separately the case where the allocations are possibly signed and the case of nonnegative allocations; the latter corresponds to the economically relevant assumption that no agents can profit from an aggregate pure loss. We offer several explicit characterizations of the inf-convolution and optimal allocations under some assumptions on the distortion functions. As another main contribution, we show that, indeed, the relative strength of risk aversion and risk seeking of the representative agents determines the structure of the optimal allocations for the original market, formalizing the aforementioned simple conjecture.

The rest of this paper is structured as follows. Section~\ref{Sec:2} formulates the risk-sharing problem that we consider in this paper. Section~\ref{Sec:3} presents an $n$-agent setting with mixed attitudes (risk-averse and risk-seeking agents), and shows how the problem reduces to a two-agent formulation, which motivates the analysis in Sections~\ref{Sec:4} and \ref{Sec:5}. Regarding constraints on the feasible allocation set, Section~\ref{Sec:4} considers general allocation sets and establishes our main existence result for optimal allocations, while Section~\ref{Sec:5} focuses on positive allocation sets. Several results extend beyond the specific risk-averse/risk-seeking setting. Section~\ref{Sec:6} integrates these results to resolve the $n$-agent problem in certain settings. Finally, Section \ref{Sec:7} concludes.

\section{Problem formulation}\label{Sec:2}

Let $\mathcal{X}$ be a convex cone of random variables on an atomless probability space $\left(\Omega,\mathcal{F},\mathbb{P}\right)$. Consider a setting with $n\in\mathbb{N}$ agents who share a random variable $X\in \mathcal{X}$,  representing the aggregate risk,  without central authority involvement. Write $[n]=\{1,\dots,n\}$. 
Positive and negative values of random variables  represent losses and surpluses, respectively.
We denote by $F_X: \mathbb{R}\mapsto [0,1]$ the cumulative distribution function of $X$, and by $S_X: \mathbb{R}\mapsto [0,1]$ the survival function of $X$.
Let $U_X$ be a uniform random variable on $[0,1]$, such that $F_X^{-1}(U_X)=X$ almost surely. The existence of $U_X$ for a general $X$ is guaranteed; see Lemma A.32 of \cite{follmer2016stochastic}.
We consider a one-period economy, where an aggregate risk $X$ is redistributed among the agents at the end of the period.
For a given $X \in \cX$ and $n \in \mathbb{N}$, the set of allocations of $X$ is
\begin{equation*}\label{def:allo}
\mathbb{A}_n(X)=\left\{(X_1, \ldots, X_n)\in\mathcal{X}^n:\sum_{i=1}^{n}X_i=X\right\}.
\end{equation*}
That is, the set of allocations of $X$ includes all possible divisions of $X$ among the $n$ agents. Note that the choice of $\mathcal{X}$ may impose constraints on the admissible allocations. In this paper, the space $\mathcal{X}$ can be chosen as the space $L^\infty$ of bounded random variables or the space $L^{+}$ of nonnegative random variables. The choice will be specified in each context.

We next impose some dependence structures on  admissible allocations.
A random vector $(X, Y)$ is said to be comonotonic (resp.~counter-monotonic) if
$(X(\omega)-X(\omega^{\prime}))(Y(\omega)-Y(\omega^{\prime})) \geq  0 ~(
        \mbox{resp.}~\le 0)$  for all $\omega, \omega^{\prime} \in \Omega$. 
An $n$-tuple $(X_{1},\ldots, X_{n})$ is called comonotonic (resp.~counter-monotonic) if each pair of its components is comonotonic (resp.~counter-monotonic).
For  results on and applications of  these concepts in actuarial science, see \cite{dhaene2002concept} on comonotonicity and \cite{dhaene1999safest} and \cite{lauzier2023pairwise} on counter-monotonicity.

The preference of each agent is represented by a \textit{distortion risk measure}
 $\rho_{h}: \mathcal{X}\mapsto \mathbb{R}$, defined as
\begin{align}\label{eq:distortion}
\rho_{h}(X)=\int X \mathrm{d}  h\circ \p =\int_0^\infty h(\p(X >  x))\, \mathrm{d} x + \int_{-\infty}^{0} (h_i(\p(X >  x))-h(1) )\, \mathrm{d} x, 
\end{align} 
where $h$ is in the set $\mathcal{H}$
of all normalized distortion functions, i.e., 
$$
\mathcal{H}=\{h:[0,1] \rightarrow [0,1] \mid h \text { is increasing, } h(0)=0 \text { and } h(1)=1\}.
$$
If $h\in \mathcal{H}$, then $\rho_h$ is also called a dual utility of \cite{yaari1987dual}. 
In the dual theory, the distortion function provides a full characterization of risk aversion: a Yarri agent is (strongly) risk-averse if and only if $h$ is concave (\cite{yaari1987dual}). Similarly, 
risk seeking corresponds to convexity of $h$.  We omit ``strong" in risk aversion and risk seeking, which is the only sense that we consider in this paper (weak risk aversion is different from strong risk aversion for distortion risk measures). 

If one drops the normalization requirement of $h(1)=1$ and monotonicity of $h$, we obtain the more general class of \textit{distortion riskmetrics} (see \cite{wang2020characterization}) $\rho_h$ with $h\in \mathcal{H}^{\mathrm{BV}}$, where
 $$
\mathcal{H}^{\mathrm{BV}}=\{h:[0,1] \rightarrow \mathbb{R} \mid h \text { is bounded variation and } h(0)=0\}.
$$
This broader class is introduced as it will be useful in the subsequent analysis.

Our goal is to study Pareto-optimal allocations among the $n$ agents in this market. An allocation $ (X_1, \dots, X_n ) \in \mathbb{A}_n(X)$ is said to be Pareto optimal in $\mathbb{A}_n(X)$ if for any $(Y_1, \dots, Y_n ) \in \mathbb{A}_n(X)$ satisfying $\rho_{h_i}(Y_i) \leq \rho_{h_i}(X_i)$ for all $i \in[n]$, we have $\rho_{i}\left(Y_i\right)=\rho_{i}\left(X_i\right)$, for all $i \in[n]$.
That is, no agent can strictly improve their position without worsening that of another.

An alternative method of finding Pareto optima is via the tool of infimal convolution. The inf-convolution $\square_{i=1}^n \rho_{h_i}$ of $n$ risk measures $\rho_{h_1}, \dots, \rho_{h_n}$ is defined as
\begin{align}\label{definition:inf-conv}
    \dsquare_{i=1}^n \rho_{h_i}(X)=\inf \left\{\sum_{i=1}^n \rho_{i}\left(X_i\right):\left(X_1, \dots, X_n\right) \in \mathbb{A}_n(X)\right\}, \quad X \in \mathcal{X}.
\end{align}
That is, the inf-convolution of $n$ risk measures is the infimum over aggregate risk values for all possible allocations.
An allocation $(X_1, \dots, X_n)$ is said to be optimal in $\mathbb{A}_n(X)$ if it attains the infimum, i.e.,
$\square_{i=1}^n \rho_{h_i}(X)=\sum_{i=1}^n \rho_{h_i} (X_i)$. It is immediate that every optimal allocation is Pareto optimal.
Moreover, when $\mathcal{X}=L^\infty$, the converse also holds and Pareto optimality coincides with optimality through inf-convolution, because of    translation invariance of the risk measures (see Proposition 1 of \cite{embrechts2018quantile}). By contrast, the converse may not hold if  $\mathcal{X}$ is chosen as $L^+$. Both cases will be considered in the subsequent analysis; see Section~\ref{Sec:4} 
and Section~\ref{Sec:5}.

A particularly useful property of distortion risk measures is the following duality. For any $X\in \mathcal{X}$, we have $\rho_{h}(-X)=-\rho_{\tilde{h}}(X)$, where $\tilde{h}(t)=1-h(1-t)$, $t\in [0,1]$, is the dual function of $h$. 
Consequently, an optimal allocation $(X_1^*, \ldots, X_n^*)$ that attains the infimum in \eqref{definition:inf-conv} also solves the utility maximization problem: 
\begin{align*}
    \sup \left\{\sum_{i=1}^n \rho_{\tilde{h}_i}\left(-X_i\right):\left(-X_1, \dots, -X_n\right) \in \mathbb{A}_n(-X)\right\}, \quad X \in \mathcal{X},
\end{align*}
and here $\rho_{\tilde{h}_i}$ can be interpreted as  agent $i$'s utility functional in the theory of \cite{yaari1987dual}.
That is, a risk sharing minimizer is simultaneously a utility maximizer under the dual distortions. In this paper, we adopt the risk sharing perspective and focus on minimizing the aggregate risk across all agents.

In a manner similar to the inf-convolution of risk measures, we define the inf-convolution of real functions, adopting  the same notation.
Given $n$ distortion functions $h_i \in \mathcal{H}$ for $i\in[n]$,  their inf-convolution $\dsquare_{i=1}^{n} h_i(x) : [0,1]\mapsto \mathbb{R}$ is defined as
\begin{align*}
    \dsquare_{i=1}^{n} h_i(x)=\inf \left\{ \sum_{i=1}^{n}h_i(x_i): x_i\in [0,1] \ \text{for}\ i\in[n]; ~  \sum_{i=1}^{n}x_i =x\right\}.
\end{align*}
Note that the $\dsquare_{i=1}^{n} h_i$ is not necessarily a distortion function since $\dsquare_{i=1}^{n} h_i(1)=1$ does not always hold. 
Throughout the paper, we write $\bigvee_{i=1}^n x_i=\max \{x_1,\cdots, x_n\}$ and $\bigwedge_{i=1}^n x_i=\min \{x_1,\cdots, x_n\}$ for any $x_1,\cdots, x_n\in \mathbb{R}$, and these operation are naturally extended to functions.

\section{Problem reduction}\label{Sec:3}
Our goal is to understand how the coexistence of risk‐averse and risk‐seeking behaviors shapes Pareto‐optimal allocations. Specifically, suppose that $n$ agents share a total risk $X$. The first $m$ are risk averse, with concave distortion functions $h_i, i\in[m]$, and the remaining $n-m$ are risk seeking with convex distortion functions $h_i, i\in [n]\setminus[m]$. 
The following theorem shows that under some mild conditions, this $n$-agent risk-sharing problem can always be reduced to an equivalent two‐agent problem.
Specifically, let $g_1=\bigwedge_{i=1}^{m}h_i$ and $g_2=\square_{j=m+1}^{n}h_j$.
These two are indeed the representative risk preferences for the risk-averse and risk-seeking groups, as shown in Theorem \ref{proposition_1}.
 The key insight of this reduction is that the collective behavior of each class can be represented by a single ``aggregate'' distortion, so the optimal allocation can be studied through a two-distortion formulation. Before moving to the result, we first introduce the following set:
    $$\mathcal{X}^{\perp}=\{X \in \mathcal{X}: \text{there exists a uniform random variable $U$ independent of} \,X\}.$$
For more discussions on $\mathcal{X}^{\perp}$, in particular its relation to $\mathcal X$,
we refer to Section 5 of \cite{liu2020inf}.

\begin{theorem}\label{proposition_1}
    Let $\mathcal{X}=L^+$. Assume that $h_i\in \mathcal{H}$ are continuous for all $i\in [n]$, with $h_i$ concave for $i\in[m]$ and $h_j$ convex for $j\in [n]\setminus[m]$. 
    For $X\in \mathcal{X}^{\perp}$, let
    $(X_1^*,\ldots, X_n^*)\in \mathbb{A}_n(X)$ be an optimal allocation. Then the following holds:
    \begin{itemize}
        \item[(i)]
        If $h_i$ is strictly concave for $i\in [m]$
        and $h_j$ is strictly convex for $j\in [n]\setminus[m]$, then the vector $(X_1^*, \ldots, X_m^*)$ is comonotonic and the vector $(X_{m+1}^*, \ldots, X_n^*)$ is counter-monotonic. 
         \item[(ii)] It holds that
\begin{align}\label{eq_proposition_1}
         \dsquare_{i=1}^{n}\rho_{h_i}(X)=
         \rho_{g_1}\dsquare \rho_{g_2}(X).
    \end{align}
    \end{itemize}
\end{theorem}

   \begin{proof}
   (i) The comonotonicity of $(X_1^*,\ldots,X_m^*)$ follows directly from the comonotonic improvement theorem of \cite{landsberger1994co}, while the counter-monotonicity of $(X_{m+1}^*,\ldots,X_n^*)$ follows from the counter-monotonic improvement theorem of \cite{lauzier2024negatively}.

   (ii)
    For $X\in \mathcal{X}^{\perp}$, 
    \begin{align}\notag
        \dsquare_{i=1}^{n}\rho_{h_i}(X)=&\inf\left\{\sum_{i=1}^{m}\rho_{h_i}(X_i)+ \sum_{i=m+1}^{n}\rho_{h_i}(X_i): (X_1, \cdots, X_n)\in \mathbb{A}_n(X)\right\}\\ \label{eq:th1}
        & =\inf\left\{\dsquare_{i=1}^{m}\rho_{h_i}(Y)+ \dsquare_{i=m+1}^{n}\rho_{h_i}(Z): (Y, Z)\in \mathbb{A}_2(X)\right\}\\ \label{eq:th2}
        &=\inf\left\{\dsquare_{i=1}^{m}\rho_{h_i}(Y')+ \dsquare_{i=m+1}^{n}\rho_{h_i}(Z'): (Y', Z')\in \mathbb{A}_2(X) ~\text{and}~ Z'\in \mathcal{X}^\perp \right\}\\  \label{eq:th3}
        &= \inf\left\{\rho_{g_1}(Y)+ \rho_{g_2}(Z): (Y, Z)\in \mathbb{A}_2(X)\right\}\\ \notag &=\rho_{g_1}\dsquare \rho_{g_2}(X).
    \end{align}
    Applying Theorem 2 of \cite{ghossoub2024counterb} to $\dsquare_{i=m+1}^{n}\rho_{h_i}(Z)$ in \eqref{eq:th1} requires $Z\in \mathcal{X}^\perp$. 
    We now show that it suffices to assume $X\in \mathcal{X}^\perp$. 
    Indeed, for any $(Y,Z)\in \mathbb{A}_2(X)$, there exists $(Y', Z')\in \mathbb{A}_2(X)$ and a standard uniform random variable $U$
    such that $(X,Y,Z)\overset{d}{=}(X,Y',Z')$ and $U$ is independent of $(X,Y',Z')$; see Lemma 2 of \cite{lauzier2024negatively}. 
    The allocation $(Y', Z')\in \mathbb{A}_2(X)$ can be constructed as follows.
    Since $X\in \mathcal{X}^{\perp}$, there exists a standard uniform random variable $U$ independent of $X$, then we can generate i.i.d. standard uniform random variables $U_1, U_2$ independent of $U$.
    Let $$Y'=F_{Y \mid X}^{-1}\left(U_1 \mid X_1\right)~\text{and}~ Z'=F_{Z \mid X,Y}^{-1}\left(U_2 \mid X, Y'\right),$$
    where $F_{Y \mid X}^{-1}(\cdot \mid y)$ is the conditional quantile function of $Y$ given $Y_1=y$ and  $F_{Z \mid X, Y}^{-1}\left(\cdot \mid x,  y\right)$ is the quantile function of $Z$ given $\left(X, Y\right)=\left(x,y\right) \in \mathbb{R}_{+}^2$. By construction, we have $\left(X, Y, Z\right) \stackrel{\mathrm{d}}{=}\left(X, Y', Z'\right)$ and $Y'+Z'=Y+Z=X$. 
    By Proposition 1 of \cite{liu2020inf}
    and uniformly continuity of $\rho_h$, the constrained inf-convolution $\dsquare_{i=m+1}^{n}\rho_{h_i}$ 
    is law-invariant; thus \eqref{eq:th2} follows. 
    Also, Theorem 3 of \cite{ghossoub2024counterb} identifies the inner inf-convolution in \eqref{eq:th3}: $$\dsquare_{i=1}^{m}\rho_{h_i}(X)=\rho_{g_1}(X)~\text{and}  ~\dsquare_{i=m+1}^{n}\rho_{h_i}(X)=\rho_{g_2}(X) ~\text{for}~ X\in \mathcal{X}^{\perp},$$ 
    with concave distortion functions $h_i$, $i \in [m]$ and 
    convex distortion functions $h_i$, $i\in [n]\setminus[m]$. 
\end{proof}

The first part  of Theorem \ref{proposition_1} also holds for the expected utility model; see \cite{lauzier2024negatively}. This can be generalized to other decision models accounting for risk-averse and risk-seeking behaviors.
In Theorem \ref{proposition_1}, $g_1$ is concave and $g_2$ is convex. Equation \eqref{eq_proposition_1} shows that, under appropriate assumptions, 
the $n$-agent problem reduces to risk sharing between a representative risk-averse agent with distortion function $g_1$ and a representative risk-seeking agent with distortion function $g_2$.
The intuition behind this result is that agents with similar risk attitudes can be treated collectively: among the risk-averse agents, comonotonic allocations are optimal, while among the risk-seeking agents, counter-monotonic allocations are optimal. Hence, the only remaining complexity lies in how the total risk is divided between these two representative agents.

It is important to note that $g_2$, being the inf-convolution of several convex distortion functions, is not necessarily a distortion function. In fact, one typically has $g_2(1)< 1$ unless all underlying  $h_i,i\in[n]$ are the identity. Let $\hat{g}_2=g_2/g_2(1)$ the normalization of $g_2$. Clearly, such normalization would not change the preferences of the agent and a Pareto-optimal allocation for agents using $\rho_{g_1}$ and $\rho_{g_2}$ is also Pareto optimal for agents using $\rho_{g_1}$ and $\rho_{\hat{g}_2}$. This observation motivates the next step of our analysis, where we examine the detailed two-agent model in detail, showing how the form of distortion functions drives the optimal risk sharing.

\section{Allocations in the whole space}\label{Sec:4}

In this section, we consider an economy with two agents who may have different risk preferences. 
As shown by Example \ref{example:fail} below, a Pareto-optimal allocation may fail to exist in the presence of risk-seeking behavior; this is a known phenomenon of risk-seeking agents.   
One possible resolution is to impose constraints on admissible allocations. 
Our first step, however, is to identify conditions on the distortion functions under which a Pareto-optimal allocation exists even without any external constraints. Throughout this section, we will work on the space $\mathcal{X}=L^\infty$. 
In Section \ref{Sec:5}, we turn to the setting with boundedness constraints on each individual allocation, where the existence of the inf-convolution is guaranteed.

\begin{example} \label{example:fail}
    Consider two agents sharing a sure amount $X=1$. Agent~1 is extremely risk-seeking with $h_1(x)=\mathbb{1}_{\{x=1\}}$, and Agent~2 is risk-averse with distortion $h_2(x)=\sqrt{x}$. Fix any event $A\in \mathcal{F}$ with $0<\mathbb{P}(A)<1$ and an arbitrary scalar $a>0$, consider the allocation 
$(X_1,X_2)=(a\,\mathbb{1}_A,\; 1-a\,\mathbb{1}_A)$.
     A direct calculation yields 
    \begin{align}\label{eq:2}
        \rho_{h_1}(a\mathbb{1}_{A}) + \rho_{h_2}(1-a\mathbb{1}_{A})=1- a\,\tilde{h}_2(\mathbb{P}(A)).
    \end{align}
    Since $\tilde{h}_2\!\big(\mathbb{P}(A)\big)>0$, the right-hand side of \eqref{eq:2} tends to $-\infty$ as $a\to\infty$. Thus, by taking unbounded bets on $A$, the total measured risk can be driven arbitrarily negative. Intuitively, any gamble of the form $a\,\mathbb{1}_A$ with $0<\mathbb{P}(A)<1$ is perceived by Agent~1 as perfectly safe, no matter how large $a$ is. The balancing position $1-a\,\mathbb{1}_A$ falls to Agent 2, whose concave distortion function evaluates this exposure ever more negatively.
    Together, this makes the pair drive the total risk down without bound. Hence, a Pareto-optimal allocation fails to exist because one side can always ``game'' the allocation by scaling up risky bets.
\end{example}

\subsection{Two general agents}
We now present our first main result, which establishes necessary and sufficient conditions for the existence of a well-defined optimal risk sharing value between two agents with heterogeneous distortion preferences.

\begin{theorem}\label{theorem:1}
    Suppose that $\mathcal{X}=L^\infty$.
    For any $h_1$, $h_2 \in \mathcal{H}$, the following are equivalent.
    \begin{enumerate}
        \item[(i)] $h_1 \dsquare h_2(1)=1$; equivalently, $h_1\ge \tilde h_2$.
        \item[(ii)] $\rho_{h_1}\dsquare \rho_{h_2}(c)=c$ for any $c\in \mathbb{R}$.
        \item[(iii)] $\rho_{h_1}\dsquare \rho_{h_2}(X)>-\infty$ holds for any $X \in \mathcal{X}$.
         \item[(iv)] $\rho_{h_1}\dsquare \rho_{h_2}(X)>-\infty$ holds for some $X \in \mathcal{X}$.
    \end{enumerate}
\end{theorem}
\begin{proof}
    (i) $\rightarrow$ (ii): 
    Note that the condition $h_1 \dsquare h_2(1)=1$ is equivalent to $h_1\geq \tilde{h}_2$ on $[0,1]$. This implies that $\rho_{h_1}(Y)\geq \rho_{\tilde{h}_2}(Y)$ for any $Y\in {L}^\infty$.
    Using translation invariance of $\rho_{h_2}$, we obtain
    \begin{align*}
        \rho_{h_1}\dsquare \rho_{h_2}(c)=
        \inf_{Y\in \mathcal{X}}\left\{\rho_{h_1}(Y)+ \rho_{h_2}(c-Y) \right\}
         = \inf_{Y\in \mathcal{X}}\left\{c+\rho_{h_1}(Y)- \rho_{\tilde{h}_2}(Y) \right\} \geq c.
    \end{align*}
    On the other hand, it is always the case that $\rho_{h_1}\dsquare \rho_{h_2}(c)\leq c$, for $c\in \mathbb{R}$. Hence, the equality follows.
    
    (ii) $\rightarrow$ (iii): 
    Take $X\in \mathcal{X}$ and $Z=X-\essinf X \geq 0$. By monotonicity of $\rho_{h_2}$, we have
    \begin{align*}
        \rho_{h_1}\dsquare \rho_{h_2}(X)&=
        \rho_{h_1}\dsquare \rho_{h_2}(Z+\essinf X)\\
        &=\inf_{Y\in \mathcal{X}}\left\{\rho_{h_1}(Y)+ \rho_{h_2}(Z+\essinf X-Y) \right\}\\
        & \geq \inf_{Y\in \mathcal{X}}\left\{\rho_{h_1}(Y)+ \rho_{h_2}(\essinf X-Y) \right\}\\ &= \rho_{h_1}\dsquare \rho_{h_2}(\essinf X)=\essinf X >-\infty.
    \end{align*}
    
    (iii) $\rightarrow$ (iv): Immediate.
    
    (iv) $\rightarrow$ (i): 
    We argue by contradiction. Assume that $h_1\dsquare h_2(1)<1$, i.e.,
    there exists $m\in [0,1]$ such that $h_1(m)<\tilde{h}_2(m)$.
    Construct a zero-sum allocation $(Y, -Y)$ with $Y=2a \mathbb{1}_{A}-a$, $a>0$ and $\mathbb{P}(A)=m$.
    Then 
    \begin{align*}
        \rho_{h_1}\dsquare \rho_{h_2}(0)&\leq \rho_{h_1}(2a \mathbb{1}_{A}-a) + \rho_{h_2}(2a \mathbb{1}_{A^c}-a)\\
        &= 2a h_1(\mathbb{P}(A))+2a h_2(1-\mathbb{P}(A))-2a \\
        &=2a (h_1(m)-\tilde{h}_2(m) ).
    \end{align*}
    As $a\rightarrow \infty$, $\rho_{h_1}\dsquare \rho_{h_2}(0)\rightarrow -\infty$.
    For any $X\in \mathcal{X}$, 
    by the fact of $X\leq \esssup X$, we have
    \begin{align*}
        \rho_{h_1}\dsquare \rho_{h_2}(X)&\leq \rho_{h_1}\dsquare \rho_{h_2}(0) + \esssup X =-\infty,
    \end{align*}
    which leads to a contradiction. Therefore, we conclude $h_1(x)\geq \tilde{h}_2(x)$ for $x\in[0,1]$, which implies that $h_1 \dsquare h_2(1)=1$.
\end{proof}

In fact, condition (i) in Theorem \ref{theorem:1} captures a dominance of caution over speculation in the perception of risk, when both risk-averse and risk-seeking behaviors coexist. This condition requires that the cautious agent's loss weights are strong enough to offset the other's optimism at each quantile level, in the sense that $h_1$ never falls below the dual $\tilde h_2$.
As a result, this condition prevents the ``gambling'' agent from driving the total risk value to $-\infty$. Without this no-free-gambling condition, the risk-sharing mechanism may fail, allowing one agent to absorb all risk without sufficient penalty.

\begin{remark}
Theorem \ref{theorem:1} yields a  criterion for when the inf-convolution of two distortion risk measures cannot admit a finite and attainable minimum.  If $h_1\dsquare h_2(1)<1$, then $\rho_{h_1}\dsquare \rho_{h_2}(X)=-\infty$ holds for any $X \in \mathcal{X}$, implying that the inf-convolution cannot be exact at $X$. 
\end{remark}

We now highlight several well-known instances that Theorem~\ref{theorem:1} recovers.

\begin{example}[Two risk-averse agents]\label{example:averse}
    If both $h_1$ and $h_2$ are concave, then condition (i) in Theorem \ref{theorem:1} is satisfied; i.e., $h_1\geq \tilde{h}_2$ since $\tilde{h}_2$ is convex. In particular, for a constant total risk, any risk-free split is Pareto optimal.
\end{example}

\begin{example}[Two risk-seeking agents]\label{example:seeking}
If both $h_1$ and $h_2$ are convex and neither is the identity, then $h_1< \tilde{h}_2$ since $\tilde{h}_2$ is concave. By Theorem \ref{theorem:1}, it follows that 
\begin{align*}
    \rho_{h_1}\dsquare \rho_{h_2}(X)=-\infty,\ \text{for}\ X\in \mathcal{X},
\end{align*}
which is consistent with the result established in Proposition 2 of \cite{ghossoub2024counterb}. 
\end{example}

\begin{example}[Inf-convolution of VaRs]\label{ex:3}
    By Corollary 2 of \cite{embrechts2018quantile}, for $\alpha_1, \alpha_2 \geq 0$, we have 
    \begin{align*}
        \mathrm{VaR}_{\alpha_1}\dsquare \mathrm{VaR}_{\alpha_2}(X)= \mathrm{VaR}_{\alpha_1+\alpha_2}(X),\ \text{for}\ X\in \mathcal{X}.
    \end{align*}
    Note that for $\alpha\geq 1$, $\mathrm{VaR}_{\alpha}(X)=-\infty$ for $X\in \mathcal{X}$. Therefore, 
    if $\alpha_1+ \alpha_2 \geq 1$, it holds that
    \begin{align}\label{eq:VaR}
        \mathrm{VaR}_{\alpha_1}\dsquare \mathrm{VaR}_{\alpha_2}(X)= -\infty,\ \text{for}\ X\in \mathcal{X}.
    \end{align}
    This conclusion also follows directly from Theorem~\ref{theorem:1}. In this case, the two agents are associated with  distortion functions  $h_1(x)=\mathbb{1}_{\left\{x>\alpha_1\right\}}$ and $h_2(x)=\mathbb{1}_{\left\{x>\alpha_2\right\}}$, respectively. Indeed, $h_1 \dsquare h_2(1)=0\neq1$ with $\alpha_1+ \alpha_2 \geq 1$, and Theorem \ref{theorem:1} then yields \eqref{eq:VaR}.
\end{example}

As a direct application of Theorem~\ref{theorem:1}, we now investigate the special case where one agent's distortion function is the dual of the other's. 
This setting is of particular interest because the pair $(h,\tilde h)$ acts as a ``mirror'' on the same risk. Valuing a loss with $h$ equals (up to sign) valuing a gain with the dual $\tilde h$.
The following proposition shows that, under this mirror setting, the inf-convolution is always well-defined, and it even reduces to a particularly simple form when  one agent is risk averse (the mirror agent is necessarily risk seeking).

\begin{proposition}\label{pro:4}
Suppose that $\mathcal{X}=L^\infty$.
For any $h\in \mathcal{H}$ and $X\in \mathcal{X}$, the following hold.
    \begin{enumerate}
        \item[(i)] $\lvert \rho_h \dsquare \rho_{\tilde{h}}(X)\rvert \leq \lVert X \rVert _\infty$.
         \item[(ii)] $\rho_h \dsquare \rho_{\tilde{h}}(X)>-\infty$.
        \item[(iii)] $\rho_h \dsquare \rho_{\tilde{h}}(c)=c$ for any constant $c\in \mathbb{R}$.
        \item[(iv)] If $h$ or $\tilde{h}$ is concave, then
    $\rho_{h}\dsquare\rho_{\tilde{h}}= \rho_{h \dsquare \tilde{h}} = \rho_{h \wedge \tilde{h}} $.
    \end{enumerate}
\end{proposition}
\begin{proof}
(i). For any $h\in \mathcal{H}$, it follows that 
    \begin{align*}
        \rho_h \dsquare \rho_{\tilde{h}}(X)=&\inf_{Y\in \mathcal{X}}\left\{\rho_h(Y)+ \rho_{\tilde{h}}(X-Y) \right\}
         = \inf_{Y\in \mathcal{X}}\left\{\rho_h(Y)- \rho_{h}(Y-X) \right\}.
    \end{align*}
    Then we have
   \begin{align*}
        \lvert \rho_h \dsquare \rho_{\tilde{h}}(X)\rvert \leq \inf_{Y\in \mathcal{X}}\left\{\lvert \rho_h(Y)- \rho_{h}(Y-X) \rvert \right\}
        \leq \lVert X \rVert _\infty.
    \end{align*}
    The second inequality holds since $\rho_h$ is lipschitz-continuous with respect to $L^\infty$ -norm.

    (ii) and (iii) follow directly from Theorem \ref{theorem:1}, since $h\dsquare\tilde{h}(1)=1$.

(iv). Without loss of generality, we assume that $h$ is concave. Then
\begin{align*}
\rho_{h}\dsquare\rho_{\tilde{h}}(X)&=\inf_{Y\in \mathcal{X}}\left\{\rho_{h}(Y)+ \rho_{\tilde{h}}(X-Y) \right\} \\
    & \geq \inf_{Y\in \mathcal{X}}\left\{\rho_{h}(Y)+ \rho_{\tilde{h}}(X)+\rho_{\tilde{h}}(-Y)  \right\}=\rho_{\tilde{h}}(X).
\end{align*}    
Also, $\rho_{h}\dsquare\rho_{\tilde{h}}\leq \rho_{\tilde{h}}$ holds. Thus, we obtain $\rho_{h}\dsquare\rho_{\tilde{h}}=\rho_{\tilde{h}}$.
A symmetric argument applies if $\tilde{h}$ is concave. 
Consequently, we obtain $\rho_{h}\dsquare\rho_{\tilde{h}}= \rho_{h \wedge \tilde{h}}$.

Next, we show that $\rho_{h \wedge \tilde{h}}=\rho_{h \square \tilde{h}}$. Note that $h \square \tilde{h}(1)=1$, and hence it suffices to show that $h \wedge \tilde{h}=h \square \tilde{h}$. Let $g(y)=h(y)-h(1-x+y)$, for each $x\in [0,1]$. Then $g(y)$ is non-decreasing in $y\in [0,x]$ due to concavity of $h$.
Therefore, $g(y)$ attains its infimum at the endpoints $y=0$. Then it follows that
\begin{align*}
    h \square \tilde{h}(x)=\inf_{0\leq y \leq x}\left\{h(y)+\tilde{h}(x-y)\right\}
    =1+\inf_{0\leq y \leq x}\left\{h(y)-h(1-x+y)\right\}=\tilde{h}(x).
\end{align*}
Similarly, if $\tilde{h}$ is concave, we have $h \square \tilde{h}(x)=h(x)$. 
\end{proof}

\subsection{One risk-averse agent and one risk-seeking agent}\label{subsec:4.2}
To approach a solution to \eqref{eq_proposition_1}, we now analyze the two-agent setting in detail. 
The following theorem identifies structural conditions on the distortion functions under which the inf-convolution admits a simple representation, including cases in which one agent is risk averse or risk seeking.

\begin{theorem}\label{theorem:no_concave_convex}
        Suppose that $\mathcal{X}=L^\infty$ and 
        $h_1, h_2 \in \mathcal{H}$. The following hold.
    \begin{itemize}
        \item[(i)] If $h_1 \wedge h_2$ is concave, then $\rho_{h_1}\dsquare \rho_{h_2}=  \rho_{h_1 \dsquare h_2}=\rho_{h_1 \wedge h_2}$.
        \item[(ii)]If $h_2$ is convex, then $\rho_{h_1}\dsquare \rho_{h_2}= \rho_{h_2} $ if and only if $h_1 \dsquare h_2(1)=1$. 
    \end{itemize}
\end{theorem}
\begin{proof}
(i)
By concavity of $h_1 \wedge h_2$, we have 
\begin{align*}
    h_1 \dsquare h_2 \geq (h_1 \wedge h_2)\dsquare (h_1 \wedge h_2)=h_1 \wedge h_2.
\end{align*}
On the other hand, it is clear that $h_1 \dsquare h_2 \leq h_1 \wedge h_2$. Hence, we conclude $h_1 \dsquare h_2 = h_1 \wedge h_2 $. Then it suffices to show $\rho_{h_1}\dsquare \rho_{h_2}= 
\rho_{h_1 \wedge h_2}$. For $X\in \mathcal{X}$ and 
$(X_1, X_2)\in \mathbb{A}_2(X)$, we have 
\begin{align*}
        \rho_{h_1}(X_1)+ \rho_{h_2}(X_2)\geq \rho_{h_1 \wedge h_2}(X_1)+ \rho_{h_1 \wedge h_2}(X_2)\geq \rho_{h_1 \wedge h_2}(X).
\end{align*}
This implies that $\rho_{h_1}\dsquare \rho_{h_2} \geq \rho_{h_1 \wedge h_2}$.
The last inequality follows from the equivalence between concavity of $h_1 \wedge h_2$ and subadditivity of $\rho_{h_1 \wedge h_2}$; see Theorem 3 of \cite{wang2020characterization}. 
Furthermore, $\rho_{h_1}\dsquare \rho_{h_2} \leq \rho_{h_1 \wedge h_2}$ holds. Therefore, the desired result is obtained.

(ii) ``If'': Note that $h_1 \dsquare h_2(1)=1$ is equivalent to $h_1\geq \tilde{h}_2$. For any $X,Y\in \mathcal{X}$,
 it follows that 
 \begin{align*}
     \rho_{h_1}(Y)+ \rho_{h_2}(X-Y)\geq \rho_{h_1}(Y)+ \rho_{h_2}(X) -\rho_{\tilde{h}_2}(Y) \geq \rho_{h_2}(X).
 \end{align*}
 The second inequality holds due to the convexity of $h_2$. 
 The above result implies that  $\rho_{h_1}\dsquare \rho_{h_2}\geq \rho_{h_2}$. Also, it is immediate that $\rho_{h_1}\dsquare \rho_{h_2}\leq \rho_{h_2}$. Hence, the equality holds.

``Only if'': Let $X=\mathbb{1}_{A}$ and $x=\mathbb{P}(A)$. 
Take an allocation $(X_1, X_2)$ of $X$ as: $X_1=\mathbb{1}_{B}$
and $X_2=\mathbb{1}_{A\setminus B}$. Write $y=\mathbb{P}(B)$. Then it follows that
\begin{align*}
    \rho_{h_1}(X_1)+\rho_{h_2}(X_2)=h_1(y)+h_2(x-y) \ \text{and}\ \rho_{h_2}(X)=h_2(x).
\end{align*}
By the condition $\rho_{h_1}\dsquare \rho_{h_2}= \rho_{h_2}$, we have $h_1(y)+h_2(x-y)\geq h_2(x)$ for any $x\in [0,1]$. Consequently,
$ h_1(y)+h_2(1-y)\geq 1$ for any $y\in [0,1]$, implying that $h_1 \dsquare h_2 (1)=1$.
\end{proof}

Theorem \ref{theorem:no_concave_convex} provides some insights for the mixed two-agent case (one risk-averse and one risk-seeking). In particular, if the existence condition $h_1 \dsquare h_2 (1)=1$ holds, i.e., the level of risk aversion dominates the level of risk seeking, then the allocation $(X_1,X_2)=(0,X)$ is optimal.
Intuitively,  the ``force'' of caution of the risk-averse agent exceeds the ``force'' of gambling of the risk-seeking agent in this setting. As a result, the efficient arrangement is that the risk-seeking agent absorbs the entire uncertain part of the risk and the risk-averse agent holds the safe position.

It is important to note that in both cases in Theorem \ref{theorem:no_concave_convex} and throughout Example \ref{example:averse}--\ref{ex:3} (with existence conditions), the identity
\begin{align}\label{equality_inf}
    \rho_{h_1}\dsquare \rho_{h_2}=\rho_{h_1\dsquare h_2}
\end{align}
always holds. 
A natural question to ask is whether this relationship still holds in more general settings, especially when  two agents' are neither risk-averse nor risk-seeking.
In fact, the equality is not universal. Example \ref{example:sub} shows that if both agents have subadditive distortion functions, then equality \eqref{equality_inf} may fail once concavity is lost.

\begin{example}[$\rho_{h_1}\dsquare \rho_{h_2}=\rho_{h_1\dsquare h_2}$ does not always hold.]\label{example:sub}
    Suppose that two agents share the same risk preference $h\in \mathcal{H}$, where $h$ is continuous and subadditive but not concave on $[0,1]$.
    Then there exists $x,y\in [0,1]$ such that 
    \begin{align}\label{ineq:loc_convex}
        2 h( x/2+ y/2)<h(x)+h(y).
    \end{align}
    Consider the total risk $X=\mathbb{1}_{A}+\mathbb{1}_{B}$ with $\mathbb{P}(A)=\mathbb{P}(B)= x/2+ y/2$ and $\mathbb{P}(A \cap B)=y$. Define an allocation $(Y,Z)$
    of $X$ with $Y=\mathbb{1}_{A}$ and $Z=\mathbb{1}_{B}$.
    By \eqref{ineq:loc_convex}, we can obtain that 
    \begin{align*}
        \rho_{h}(Y)+ \rho_{h}(Z)<\rho_{h}(X)=h(\p(A \cup B))+h(\p(A \cap B)).
    \end{align*}
    Therefore, for such $X$, we have $\rho_{h}\dsquare \rho_{h}(X)<\rho_{h}(X)=\rho_{h\dsquare h}(X)$. The equality holds due to subadditivity of $h$, implying $h\dsquare h=h$. Below we give a concrete example.
    Let $h(x)=\max \left\{\sqrt{x}, 2x-x^2\right\}$ for $x\in [0,1]$. Clearly, both $\sqrt{x}$ and $2x-x^2$ are subadditive, hence
     $h(x)$ is subadditive. Consider $X=\mathbb{1}_{A}+\mathbb{1}_{B}$ with $\mathbb{P}(A)=\mathbb{P}(B)=0.38$ and $\mathbb{P}(A \cap B)=0.26$.
     Then we can calculate 
     $$
     \rho_{h}(Y)+ \rho_{h}(Z)=2h(0.38)= 1.233<\rho_{h}(X)=h(0.5)+ h(0.26)=1.260.
     $$    
\end{example}

Whereas the results in this section are presented for generality on the space $\mathcal{X}=L^\infty$, the subsequent sections will focus on $\mathcal{X}=L^+$, the set of nonnegative random variables, where only nonnegative allocations are admissible to ensure well-defined risk-sharing problems.

\section{Nonnegative allocations}\label{Sec:5}

Without any constraints on the allocation set, an optimal allocation may fail to exist under certain conditions, for instance, when $h_1(x)< \tilde{h}_2(x)$ on $[0,1]$, as shown in Theorem \ref{theorem:1}. In this section, we consider the setting  $\mathcal{X}=L^{+}$; that is, both the aggregate risk and the admissible allocations are required  to be nonnegative. 
This setting is economically intuitive: it means that for each agent, there cannot be
any profit from an aggregate pure loss. It is a natural assumption in many applications, such as peer-to-peer insurance.

Even with the nonnegativity constraint, solving $\rho_{h_1}\dsquare\rho_{h_2}$ with heterogeneous preferences can be challenging. 
In this section, we   proceed in two directions:
(i) we study the inf-convolution for specific structural classes of distortion functions $h_1$ and $h_2$; 
(ii) we evaluate $\rho_{h_1}\dsquare\rho_{h_2}(X)$ for tractable classes of risks $X$ (e.g., Bernoulli-type random variables) without specifying the distortion type. 
These studies also include, as notable special cases, settings with one risk-averse and one risk-seeking agent.

\subsection{Two agents with special distortion functions}
In analyzing risk sharing with specific preference classes, we focus on the relationship between the constrained inf-convolution $\rho_{h_1}\dsquare\rho_{h_2}$ and the benchmark $\rho_{h_1\dsquare h_2}$. We identify when this benchmark provides an upper bound for $\rho_{h_1}\dsquare\rho_{h_2}$ and when this bound is attained. Using $\rho_{h_1\dsquare h_2}$
as the benchmark is natural since, in all known tractable settings such as risk-averse/risk-seeking pairs and quantile-based specifications, one has the identity \eqref{equality_inf}.

An assumption of monotonicity of optimizers for $h_1\dsquare h_2$ would be important in the subsequent analysis.
\begin{assumptionp}{COIN}\label{assump:coin}
    Let $h_1,h_2\in \mathcal{H}$.
    There exists an increasing function $f$ with $x-f(x)$  increasing  such that $h_1(f(x))+h_2(x-f(x))=h_1 \dsquare h_2(x)$.
\end{assumptionp}
Clearly, Assumption \ref{assump:coin} ensures that the optimal split of $x$ moves monotonically as $x$ increases; equivalently, there exists an increasing selector $f$ with a co-increasing residual $x\mapsto x-f(x)$.
In particular, when both $h_1$ and $h_2$ are strictly convex, such selector exists; see Lemma 2 of \cite{ghossoub2024counterb}. 
\begin{remark}
    The importance of Assumption~\ref{assump:coin} arises precisely when one attempts to achieves the benchmark value $\rho_{h_1\dsquare h_2}(X)$ through an explicit allocation construction. To attain $\rho_{h_1\dsquare h_2}(X)$, one would naturally construct an allocation $(X_1, X_2)$ of $X$ so that the tail of each component corresponds to the optimal split of the tail of $X$. This requires that the splitting function $f$, which determines how the tail mass of $X$ is divided between $X_1$ and $X_2$, moves monotonically with the tail level. Assumption~\ref{assump:coin} ensures that the benchmark allocation is well-defined.
\end{remark}
The following theorem shows that under Assumption \ref{assump:coin}, the constrained inf-convolution $\rho_{h_1}\dsquare\rho_{h_2}$ can always be bounded by $\rho_{h_1\dsquare h_2}$.
\begin{theorem}\label{theorem:coin}
    Suppose that Assumption \ref{assump:coin} holds and $\mathcal{X}=L^+$. For $X\in \mathcal{X}$, it holds that 
    \begin{align}\label{eq:coin}
         \rho_{h_1\dsquare h_2}\dsquare\rho_{h_1\dsquare h_2}(X)\leq\rho_{h_1}\dsquare\rho_{h_2}(X)\leq  \rho_{h_1\dsquare h_2}(X).
     \end{align}
\end{theorem}
\begin{proof}
    Since $h_1\dsquare h_2\leq \min \left\{h_1, h_2\right\}$, the monotonicity of  distortion risk measures yields 
    the first inequality. It remains to show that $\rho_{h_1}\dsquare\rho_{h_2}(X)\leq  \rho_{h_1\dsquare h_2}(X)$.
    For $X\geq 0$ and $n\in \mathbb{N}$, define 
    \begin{align}\label{eq:xn}
        X_n=\sum_{i=0}^{n 2^n-1} \frac{i}{2^n}\mathbb{1}_{\left\{\frac{i}{2^n}<X \leq \frac{i+1}{2^n}\right\}}+ n \mathbb{1}_{\left\{X \geq n\right\}},
    \end{align}
    so that $X_n\uparrow X$. Equivalently, \eqref{eq:xn}
    can be formulated as:
    \begin{align*}
        X_n=\frac{1}{2^n} \sum_{k=1}^{n 2^n} \mathbb{1}_{A_{n}^k}, ~\text{where}~ A_{n}^k=\left\{X\geq \frac{k}{2^n}\right\}, k\in [2^n n].
    \end{align*}
    Clearly, $A_{n}^k\subseteq A_{n}^{{k-1}}$ for $k\in [2^n n]\setminus[1]$. Hence, $(\mathbb{1}_{A_{n}^1}, \cdots,\mathbb{1}_{A_{n}^{2^n n}})$ is comonotonic. 
    Write $p_{n}^k=\mathbb{P}(A_{n}^k)$ and $\mathcal{S}_n^k=\argmin_{0\leq t\leq p_{n}^k} \left\{h_1(t)+h_2(p_{n}^k-t)\right\}$ for $k\in [2^n n]$.
    Let
    $u_{n}^k \in \mathcal{S}_n^k$ be the selector within Assumption \ref{assump:coin} such that
     $u_{n}^k$ and $p_{n}^k-x_{n}^k$ are non-increasing in $k$. 
    By comonotonic additivity of $\rho_{h_1 \dsquare h_2}$, 
    it follows that
    \begin{align*}
        \rho_{h_1 \dsquare h_2}(X_n)&=\frac{1}{2^n} \sum_{k=1}^{n 2^n} \rho_{h_1 \dsquare h_2}(\mathbb{1}_{A_{n}^k})=\frac{1}{2^n} \sum_{k=1}^{n 2^n}{h_1 \dsquare h_2}(p_{n}^k)\\
        &=\frac{1}{2^n} \sum_{k=1}^{n 2^n}(h_1(u_{n}^k)+h_2(p_{n}^k-u_{n}^k)).
    \end{align*}
    Our goal is to construct an allocation $(Y_n, Z_n)\in A_2(X_n)$ such that $\rho_{h_1}(Y_n)+ \rho_{h_1}(Z_n)=\rho_{h_1 \dsquare h_2}(X_n)$ for each $n \in \mathbb{N}$.
    Define the sets
    \begin{align*}
        & B_{n}^k=B_{n}^{k+1} \cup \left\{
        F_{X_n} \left(\frac{k}{2^n}\right)\leq U_{X_n}< F_{X_n} \left(\frac{k}{2^n}\right)+(u_{n}^k-u_{n}^{k+1})\right\}, ~ k=1, \cdots, 2^n n-1,\\
        & B_{n}^{2^n n}=\left\{F_{X_n}(n)\leq U_{X_n}< F_{X_n}(n)+u_{n}^{2^n n}\right\}.
    \end{align*}
    Then for all $k\in [2^n n]$, we have
    $
    \mathbb{P}(B_{n}^k)=u_{n}^k.
    $
     It can also be verified that 
     $B_{n}^k\subseteq B_{n}^{k-1}$ and $A_{n}^k\backslash B_{n}^k\subseteq A_{n}^{k-1}\backslash B_{n}^{k-1}$ for $k\in [2^n n]\setminus[1]$.
    Define the allocation
     \begin{align}\label{eq:allocation}
         Y_n=\frac{1}{2^n} \sum_{k=1}^{n 2^n} \mathbb{1}_{B_{n}^k}~ \text{and}~ Z_n=\frac{1}{2^n} \sum_{k=1}^{n 2^n} \mathbb{1}_{A_{n}^k\backslash B_{n}^k}.
     \end{align}
     Clearly, $Y_n+Z_n=X_n$ for each $n$.
     Moreover, comonotonic additivity leads to 
     \begin{align}\notag
         \rho_{h_1}(Y_n)+\rho_{h_2}(Z_n)
         &=\frac{1}{2^n} \sum_{k=1}^{n 2^n}\big(h_1(\mathbb{P}(B_n^k))+h_2(\mathbb{P}(A_n^k)-\mathbb{P}(B_n^k))\big)
         \\ \notag
         &=
         \frac{1}{2^n} \sum_{k=1}^{n 2^n}\big({h_1}(u_{n}^k)+h_2(p_{n}^k-u_{n}^k)\big)\\ \label{eq:n}
         &=
         \frac{1}{2^n}\sum_{k=1}^{n 2^n} h_1\dsquare h_2(p_{n}^k)=\rho_{h_1\dsquare h_2}(X_n).
     \end{align}
     Consequently, it follows that
     \begin{align*}
         \rho_{h_1}\dsquare \rho_{h_2}(X)&=
         \lim_{n \rightarrow \infty}\rho_{h_1}\dsquare \rho_{h_2}(X_n)
         \\
         &\leq 
         \lim_{n \rightarrow \infty}\big(\rho_{h_1}(Y_n)+\rho_{h_2}(Z_n)\big)=\lim_{n \rightarrow \infty} \rho_{h_1 \dsquare h_2}(X_n)=\rho_{h_1 \dsquare h_2}(X).
     \end{align*}
     The first equality holds since $\rho_{h_1}\dsquare \rho_{h_2}$ is continuous from below; see Lemma \ref{lemma:2_a} in Appendix \ref{appendix:A}
     for more details. Therefore, the result implies that $\rho_{h_1}\dsquare\rho_{ h_2}(X)\leq
    \rho_{h_1\dsquare h_2}(X)$.
\end{proof}

As we can see, Theorem \ref{theorem:coin} is non-constructive and it does not specify an allocation that attains the upper bound $\rho_{h_1\dsquare h_2}(X)$. On an atomless space, however, when there exists a uniform random variable independent of $X$, the attainable allocation is available; see Lemma 3 of \cite{lauzier2024negatively}.

As noted earlier, Assumption \ref{assump:coin} holds whenever both $h_1$ and $h_2$ are convex; however, convexity is not required. There exists nonconvex pairs $(h_1,h_2)$ that still admits such an increasing minimizer, as demonstrated in Example \ref{example:concave_convex} involving a risk-averse agent and a risk-seeking agent. Therefore, Theorem \ref{theorem:coin} applies strictly beyond the convex setting.

\begin{example}\label{example:concave_convex}
 Let $g_1(x)=\max \left\{0, 4/3(x-1/4)\right\}$ and 
 $g_2(x)=\min \left\{ 8/7x, 1\right\}$  
 for $x\in [0,1]$.
 Clearly, $g_1$ is convex and $g_2$ is concave with $g_2\leq \tilde{g}_1$. 
 Then we can calculate 
 $
 g_1 \dsquare g_2(x)=\max \left\{ 0, 8/7(x-1/4)\right\}.
 $
  Define $f(x)=\min \left\{x, 1/4\right\}$, thus $x-f(x)=\max \left\{0, x-1/4\right\}$. Both $f(x)$ and $x-f(x)$ are non-decreasing; see Figure \ref{fig:example3} for a detailed illustration. 
     It can be verified that 
      $g_1(f(x))+g_2(x-f(x))=g_1\dsquare g_2(x)$ for $x\in [0,1]$.
      
    \begin{figure}[ht]
     \centering
     \begin{subfigure}[b]{0.48\textwidth}
         \centering
        \includegraphics[width=7.5cm]{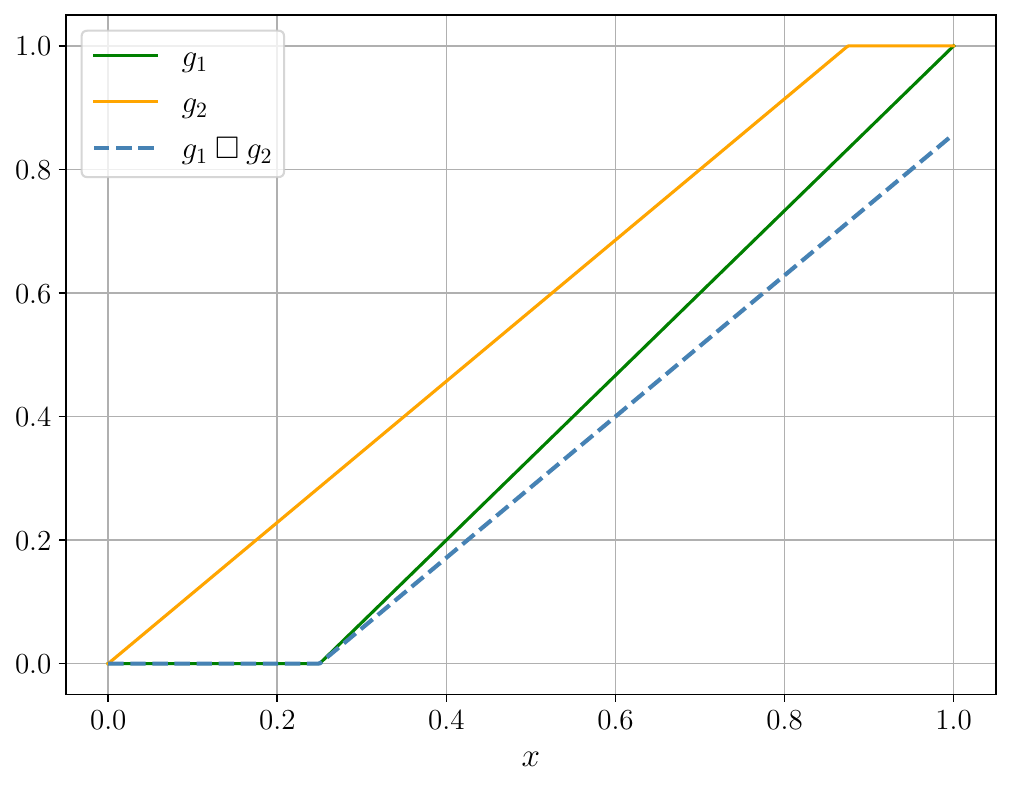}
         \caption{$g_1, g_2$ and $g_1 \dsquare g_2$}
         \label{fig:example3_1}
     \end{subfigure}
     \hfill
     \begin{subfigure}[b]{0.48\textwidth}
         \centering
        \includegraphics[width=7.5cm]{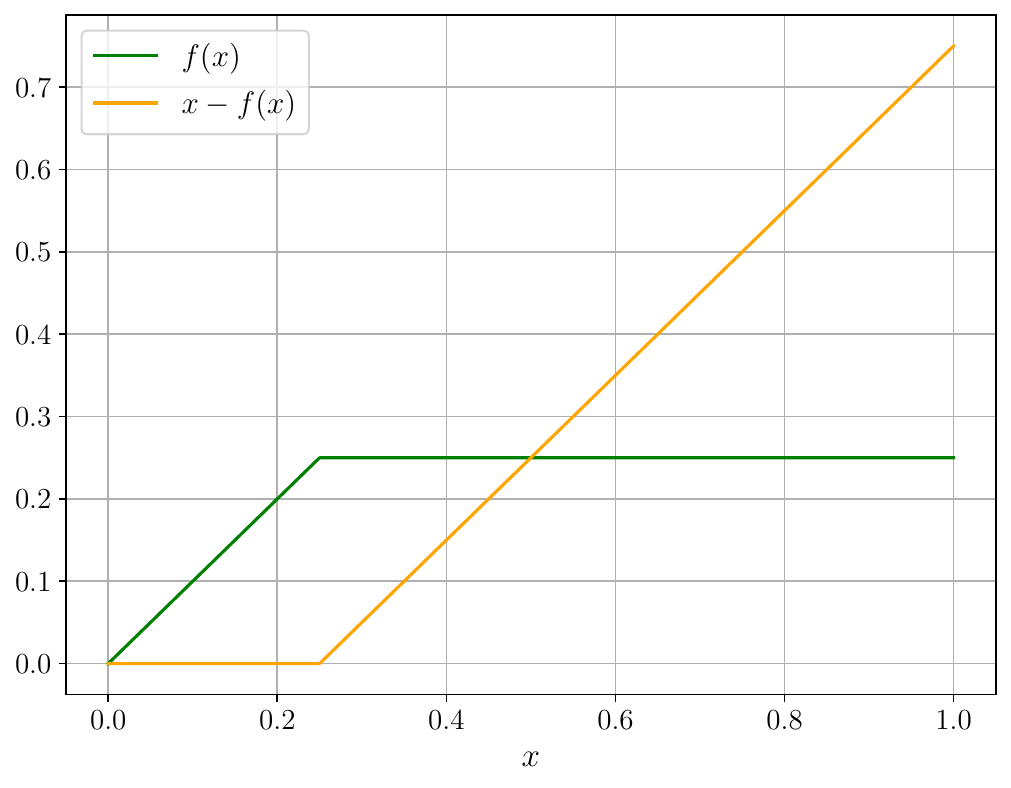}
         \caption{Optimizer for $g_1 \dsquare g_2$ }
         \label{fig:example3_2}
     \end{subfigure}
        \caption{An illustration of $g_1 \dsquare g_2$ with $a=1/4$ and $b=7/8$}
        \label{fig:example3}
\end{figure}
\end{example}

In fact, the second inequality in \eqref{eq:coin} becomes an equality when the two agents' distortion functions are of the type described in Example~\ref{example:concave_convex}. 
Before stating that result, we provide a technical lemma that will be used in the sequel, which analyzes risk sharing between a VaR agent and a distortion agent. 
The result also generalizes Theorem 5.3 of \cite{wang2020characterizing}.
To make the statement more precise, we introduce the following notation. For $h\in \mathcal{H}$, let  $\alpha(h)=\sup \left\{t\in [0,1]: h(t)=0\right\}$. The function 
$$\hat{h}(t)=h((t+\alpha(h))\wedge 1), \ \ t\in [0,1]$$ 
is called the active part of $h$. Additionally, let $h^{a}(t)=h((t-a)_{+})$ for $t\in [0,1]$ and $a\in [0,1]$.

\begin{lemma}\label{lemma:var}
    Suppose that $h\in \mathcal{H}$ and $\mathcal{X}=L^+$. For any $\alpha\in [0,1]$ and 
    $X\in \mathcal{X}$, we have
\begin{align}\label{eq:var_dis}
\mathrm{VaR}_{\alpha}\dsquare \rho_{h}(X)= \rho_{h^{\alpha}}(X).
    \end{align}
\end{lemma}

\begin{proof}
    We first show that $\mathrm{VaR}_\alpha \square \rho_h(X) \geq \rho_{h^\alpha}(X)$. For any allocation $(Y, X-Y)$ of $X$ with $0\leq Y\leq X$, it follows that
    \begin{align*}
\mathrm{VaR}_\alpha(Y)+\rho_h(X-Y)&=\mathrm{VaR}_\alpha(Y)+\int_0^1\mathrm{VaR}_{\beta}(X-Y) \mathrm{d}h(\beta)\\&=\int_0^{1-\alpha}\mathrm{VaR}_\alpha(Y)+\mathrm{VaR}_{\beta}(X-Y) \mathrm{d}h(\beta)\\&\quad +
         \int_{1-\alpha}^{1}\mathrm{VaR}_\alpha(Y)+\mathrm{VaR}_{\beta}(X-Y) \mathrm{d}h(\beta)\\
        &\geq \int_0^{1-\alpha}\mathrm{VaR}_{\alpha+\beta}(X) \mathrm{d}h(\beta)
        \\&= \int_0^{1}\mathrm{VaR}_{\beta'}(X) \mathrm{d}h((\beta'-\alpha)_{+})=\rho_{h^\alpha}(X).
    \end{align*}
    The inequality is due to $0\leq Y\leq X$ and 
    $\mathrm{VaR}_{\alpha+\beta}(X_1+X_2)\leq \mathrm{VaR}_{\alpha}(X_1)+\mathrm{VaR}_{\beta}(X_2)$ for $X_1, X_2\in \mathcal{X}$; see Corollary 1 of \cite{embrechts2018quantile}.
    
    Next we show that  $\mathrm{VaR}_\alpha \square \rho_h(X) \leq \rho_{h_\alpha}(X)$. Note that $\mathrm{VaR}_\alpha\left(X \mathbb{1}_{\left\{U_X>1-\alpha\right\}}\right)=0$. By straightforward calculation, we have
    \begin{align*}
    \mathrm{VaR}_\alpha \square \rho_h(X)&\leq
     \mathrm{VaR}_\alpha\left(X \mathbb{1}_{\left\{U_X>1-\alpha\right\}}\right)+\rho_h\left(X \mathbb{1}_{\left\{U_X \leq 1-\alpha\right\}}\right)\\
&=  \int_0^{\infty} h\left( \mathbb{P}\left(X \mathbb{1}_{\left\{U_X \leq 1-\alpha\right\}}>t\right)\right) \mathrm{d} t \\
&= \int_0^{\infty} h \left(\mathbb{P}\left(\left\{ X > t \right\} \cap \left\{U_X \leq 1-\alpha\right\}\right) \right)\mathrm{d} t \\
&=\int_0^{\infty} h \left((\p(X>t)-\alpha)_{+}\right)\mathrm{d} t =\rho_{h^\alpha}(X).
   \end{align*}
   Therefore, the desired result is obtained.
\end{proof}

\begin{remark}
    For any $\alpha\in [0,1]$, we denote by $g^{\alpha}$ the distortion function of $\mathrm{VaR}_{\alpha}$, i.e., $g^{\alpha}(x)=\mathbb{1}_{\left\{x> \alpha\right\}}$. Clearly, $g^{\alpha}\dsquare h(x)=h^{\alpha}(x)$ for $x\in[0,1]$. Therefore, \eqref{eq:var_dis} in fact states that
$\rho_{g^{\alpha}}\dsquare \rho_{h}(X)=\rho_{g^{\alpha} \dsquare h}(X)$ for $X\in \mathcal{X}$. 
\end{remark}

We now investigate the structure of optimal allocations in a two-agent risk sharing setting involving one risk-averse and one risk-seeking participant, where both agents are characterized by piecewise linear distortion functions, as generalized from Example~\ref{example:concave_convex}.
In the next proposition, the distortion  functions $g_1$ and $g_2$ correspond to a left-tail Expected Shortfall (see \cite{embrechts2015aggregation}) and an Expected Shortfall, respectively. 

\begin{proposition}\label{theorem:pie_linear}
    Suppose that $\mathcal{X}=L^+$. Let $h_1(x)=\max \left\{0, (x-a)/(1-a)\right\}$ and 
$h_2(x)=\min \left\{ x/b, 1\right\}$  
for $x\in [0,1]$, where $a,b\in(0,1)$ and $a+b\geq 1$.
Then  $$
    \rho_{h_1}\dsquare \rho_{h_2}(X)=\rho_{h_1\dsquare h_2}(X) = \rho_{h} (X) \mbox{~~~for all $X\in \mathcal{X}$},
$$
where  $$
       h(x)= \max \left\{ 0, (x-a)/b\right\} , ~~~x\in [0,1].
$$
Moreover, an optimal allocation $(X_1, X_2)$
of $X$ is given by 
\begin{align*}
    X_1=X\mathbb{1}_{\left\{U_X\geq 1-a\right\}}, ~\text{and}~ X_2=X\mathbb{1}_{\left\{U_X< 1-a\right\}}.
\end{align*}
\end{proposition}

\begin{proof}
    Clearly, $\alpha(h_1)=a$ and $\hat{h}_1(x)=\min \left\{ x/(1-a), 1\right\}$.
    We can calculate 
\begin{align}\label{eq:inf_g1g2}
       h_1 \dsquare h_2(x)=\max \left\{ 0, (x-a)/b\right\}=h_2^a(x)= h(x), ~~~x\in [0,1].
    \end{align}
    By Lemma \ref{lemma:var}, it follows that
    \begin{align*}
\operatorname{VaR}_{a}\dsquare \rho_{\hat{h}_1}(X)=\rho_{h_1}(X).
    \end{align*}
  Note that by using Lemma 2 of \cite{liu2020inf},  we can see that inf-convolutions are associative. Therefore, 
    \begin{align*}
        \rho_{h_1}\dsquare \rho_{h_2}(X)&= (\operatorname{VaR}_{a}\dsquare \rho_{\hat{h}_1}) \dsquare\rho_{h_2}(X)
        \\&=
        \operatorname{VaR}_{a}\dsquare (\rho_{\hat{h}_1} \dsquare \rho_{h_2})(X)
       =\operatorname{VaR}_{a}\dsquare \rho_{h_2}(X)=\rho_{h_1\dsquare h_2}(X),
    \end{align*}
    where the last equality follows from Lemma \ref{lemma:var} and \eqref{eq:inf_g1g2}.

Next, we find the optimal allocation. 
    Let $f(x)=\min \left\{x, a\right\}$, so that $x-f(x)=\max \left\{0, x-a\right\}$. Both $f(x)$ and $x-f(x)$ are non-decreasing. 
    It can be verified that 
     $h_1(f(x))+h_2(x-f(x))=h_1\dsquare h_2(x)$ for $x\in [0,1]$.
     Then it follows that
\begin{align*}
    \rho_{h_1}(X_1)&=\int_0^\infty h_1(\p(X\mathbb{1}_{\left\{U_X\geq 1-a\right\}}>t))\mathrm{d}t\\
    &=\int_0^\infty h_1(\p(\left\{X>t\right\}\cap \left\{U_X\geq 1-a\right\}))\mathrm{d}t\\
    &=\int_0^\infty h_1(\p(U_X\geq \max \left\{F_X(t), 1-a\right\}))\mathrm{d}t\\
    &=\int_0^\infty h_1(\min \left\{a, S_X(t)\right\})\mathrm{d}t=\int_0^\infty h_1(f(S_X(t)))\mathrm{d}t.
\end{align*}
Similarly, we have $ \rho_{h_2}(X_2)=\int_0^\infty h_1(S_X(t)-f(S_X(t)))\mathrm{d}t$.
This implies that $\rho_{h_1}(X_1)+\rho_{h_2}(X_2)=\rho_{h_1\dsquare h_2}(X)$.
\end{proof}

Although Proposition~\ref{theorem:pie_linear} covers only a subset of the mixed (one risk-averse and one risk-seeking) cases, it extends part~(ii) of Theorem \ref{theorem:no_concave_convex} by removing the constraint of $h_1\square h_2(1)=1$. 
A further implication of Proposition~\ref{theorem:pie_linear} is that comonotonic allocations cannot be optimal unless the parameters satisfy $a+b=1$.
In particular, when the degree of risk seeking exceeds that of risk aversion, counter-monotonic allocations would strictly outperform comonotonic ones, as stated in Proposition \ref{proposition:no_como}.

\begin{proposition}\label{proposition:no_como}
    Suppose that $\mathcal{X}=L^+$,  $h_1$ is concave, and $h_2$ is convex. Then
    \begin{itemize}
        \item[(i)] $h_1\dsquare h_2=h_2$ if and only if $h_1\dsquare h_2(1)=1$.
        \item[(ii)]If $h_1 \dsquare h_2(1)<1$, then a comonotonic allocation of $X\in \mathcal{X}$ is never optimal.
    \end{itemize}
\end{proposition}
\begin{proof}
    (i) ``Only if'' part is trivial. 
    We only show ``if'' part. 
    The condition of $h_1\dsquare h_2(1)=1$ is equivalent with $h_1\geq\tilde{h}_2$. For any $x\in [0,1]$ and $y\in[0,x]$, we have 
    \begin{align*}
        h_1(y)\geq 1-h_2(1-y)\geq h_2(x)-h_2(x-y).
    \end{align*}
    The second inequality holds due to convexity of $h_2$. Thus, this implies $h_1\dsquare h_2=h_2$.

    (ii) If there exists a comonotonic optimal allocation, then $\rho_{h_1}\dsquare \rho_{h_2}(X)=\rho_{h_2}(X)$. This implies that $h_1\dsquare h_2=h_2$, and hence $h_1 \dsquare h_2(1)=1$ by result of (i), contradicting the assumption that $h_1 \dsquare h_2(1)<1$.
\end{proof}

Proposition~\ref{proposition:no_como} shows that in the presence of a sufficiently strong risk-seeking (i.e., when $h_1\dsquare h_2(1)<1$), comonotonic allocations fail to be optimal.
Example~\ref{example:1} below illustrates highlights how this effect arises under power-function distortion functions.

\begin{example}\label{example:1}
Set $h_1(x)=1-(1-x)^2$ and $h_2(x)=x^3$ for $x\in [0,1]$ and let $X \sim \text{Bernoulli}(p)$ for $p\in [0,1]$. For any comonotonic allocation $(X_1, X_2)$ of $X$, it holds that
\begin{align*}
    \rho_{h_1}(X_1)+\rho_{h_2}(X_2)\geq \rho_{h_1\wedge h_2}(X)= h_2(p).
\end{align*}
By straightforward calculation, we have 
\begin{align*}
    h_1 \dsquare h_2 (x)= h_2(x) \mathbb{1}_{\left\{x \leq \sqrt{\frac{2}{3}}\right\}} + (h_1(f(x))+ h_2(x-f(x)))\mathbb{1}_{\left\{x >\sqrt{\frac{2}{3}}\right\}},
\end{align*}
where $f(x)=(3x-1-\sqrt{7-6x})/3$. It can be seen from Figure \ref{fig:example_1} that $h_1 \dsquare h_2 (x) < h_2(x)$ for $x>\sqrt{2/3}$. 
Take $X_1'=\mathbb{1}_{A}$ and $X_2'=\mathbb{1}_{B}$, where $A$ and $B$ are disjoint,  
$\mathbb{P}(A \cup B)= p$ and satisfies $h_1(\mathbb{P}(A))+h_2(\mathbb{P}(B))=h_1 \dsquare h_2 (p)$.
Clearly, $(X_1', X_2')$ is a counter-monotonic allocation of $X$. For $p >\sqrt{2/3}$, we can derive that
\begin{align*}
    \rho_{h_1}(X_1')+\rho_{h_2}(X_2')=h_1 \dsquare h_2 (p) <h_2(p).
\end{align*}
Hence, a comonotonic allocation $(X_1, X_2)$ is never optimal.
    \begin{figure}[ht]
    \centering
    \includegraphics[width=8.5cm]{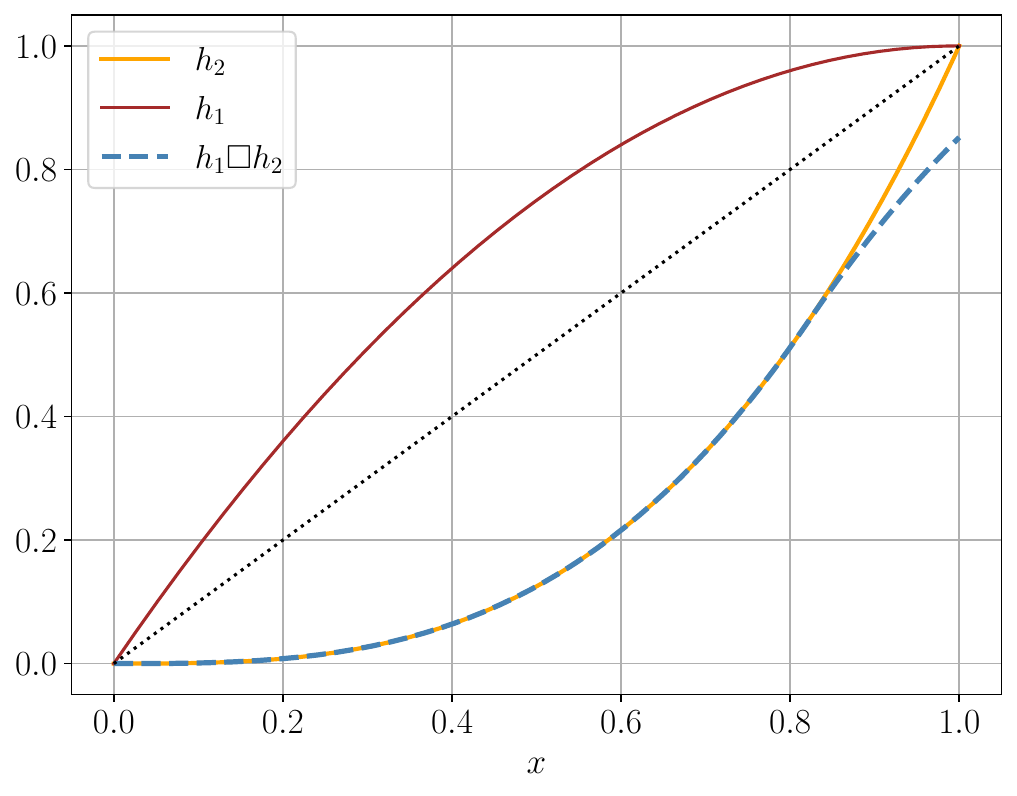}
    \caption{An illustration of Example \ref{example:1}.}
    \label{fig:example_1}
    \end{figure} 
\end{example}

\subsection{Two agents sharing a Bernoulli-type risk}
In this subsection, we consider a simple setting of two agents sharing a Bernoulli random risk with unrestricted risk preferences. In this binary setting, the inf-convolution and optimal allocations admit closed-form solutions. Furthermore, we explore how these insights extend to richer risk distributions. 

\begin{theorem}\label{theorem:ele}
    Suppose that $h_1, h_2\in\mathcal{H}$ and 
    $\mathcal{X}=L^+$. For any $A\in \mathcal{F}$ and $a\in \mathbb{R}_{+}$, it holds that
\begin{align}\label{eq:thm_indi}
        \rho_{h_1}\dsquare \rho_{h_2}(a\mathbb{1}_{A})=\rho_{h_1\dsquare h_2}(a\mathbb{1}_{A})=
        a h_1\dsquare h_2(\p(A)).
    \end{align}
    Moreover, a Pareto-optimal allocation is given by $X_1=a\mathbb{1}_{B}$ and $X_2=a\mathbb{1}_{A\setminus B}$ satisfying $h_1(\mathbb{P}(B))+h_2(\mathbb{P}(A\setminus B))=h_1 \dsquare h_2(\mathbb{P}(A))$.
\end{theorem}
\begin{proof}
By the positive homogeneity of $\rho_{h_1\dsquare h_2}$, it is trivial to verify the second equality in \eqref{eq:thm_indi}.
To show the first equality, we first state that $ \rho_{h_1}\dsquare \rho_{h_2}(a\mathbb{1}_{A})\geq \rho_{h_1\dsquare h_2}(a\mathbb{1}_{A})$.
With $0\leq Y \leq a \mathbb{1}_{A}$, it follows that
    \begin{align*}
        \rho_{h_1}(Y)+ \rho_{h_2}(a \mathbb{1}_{A}-Y) &= \int_0^\infty h_1(\mathbb{P}(Y >  t)) \mathrm{d} t +\int_0^\infty h_2(\mathbb{P}(a \mathbb{1}_{A}-Y>t)) \mathrm{d} t \\
        &= \int_0^a h_1(\mathbb{P}(Y\ge t)) \mathrm{d} t + \int_0^a h_2(\mathbb{P}(a-Y>t|A)\mathbb{P}(A)) \mathrm{d} t\\
        &= \int_0^a h_1(\mathbb{P}(Y\ge t)) \mathrm{d} t + \int_0^a h_2(\mathbb{P}(Y<t|A)\mathbb{P}(A)) \mathrm{d} t\\
        &= \int_0^a h_1(\mathbb{P}(Y\ge t)) \mathrm{d} t + \int_0^a h_2(\mathbb{P}(\{Y<t\}\cap A) ) \mathrm{d} t\\
        & \geq \int_0^a h_1 \dsquare h_2(\mathbb{P}(A)) \mathrm{d} t = a h_1 \dsquare h_2(\mathbb{P}(A)),
    \end{align*}
    where we used $\mathbb{P}(Y\ge t) + \mathbb{P}(\{Y<t\}\cap A) \ge \p(A).$
     Thus, we have $ \rho_{h_1}\dsquare \rho_{h_2}(\mathbb{1}_{A}) \geq a h_1 \dsquare h_2(\mathbb{P}(A))$.

    Next, we show the converse direction. Take an allocation $(X_1, X_2)$ of $X$ as $X_1=a \mathbb{1}_{B}$  and $X_2=a\mathbb{1}_{C}$, where $B \cup C= A$, $B \cap C=\varnothing$ and 
    $h_1(\mathbb{P}(B))+h_2(\mathbb{P}(C))=a h_1 \dsquare h_2(\mathbb{P}(A))$, 
    we have
    $\rho_{h_1}(X_1)+ \rho_{h_2}(X_2)= a h_1 \dsquare h_2(\mathbb{P}(A))$. Consequently, $ \rho_{h_1}\dsquare \rho_{h_2}(a \mathbb{1}_{A}) \leq \rho_{h_1\dsquare h_2}(a\mathbb{1}_{A})$ holds for any $a\in \mathbb{R}_{+}$.
\end{proof}

Notably, Theorem \ref{theorem:ele} does not rely on any concavity or convexity assumptions about the distortion functions. Moreover, it characterizes how the total probability mass $\p(A)$ is divided between the two agents so as to minimize the total risk value, with the splitting probabilities determined by 
$h_1\dsquare h_2$.
To gain intuition, we provide an example to see how the optimal split changes when one agent is risk averse and the other is risk seeking, particularly in cases where the dominance condition $h_1\geq \tilde{h}_2$ fails.
Example \ref{example:6} illustrates this behavior explicitly, showing how the share of risk held by each agent varies with the total probability level $\p(A)$.

\begin{example}\label{example:6}
    Suppose that $h_1(x)=1-(1-x)^{\alpha}$ and $h_2(x)=x^\beta$ with $1<\alpha < \beta$. Clearly, $h_1$ is concave,  $h_2$ is convex, and $h_1 \leq \tilde h_2$. Let $p_0=(\alpha/\beta)^{\frac{1}{\beta-1}}$ and 
    $(\mathbb{1}_{B}, \mathbb{1}_{A\setminus B})$ be an optimal allocation of $\mathbb{1}_{A}$.
    We can show that 
    \begin{enumerate}
    \vspace{-1mm}
        \item[(i)] If $\mathbb{P}(A)\leq p_0$,  then $\mathbb{P}(B)=0$ and $\mathbb{P}(A\backslash B)=\mathbb{P}(A)$, so the risk-seeking agent bear  the entire risk. \vspace{-3mm}
        \item[(ii)] If $\mathbb{P}(A)> p_0$, then as $\mathbb{P}(A)$ increases, 
    $\mathbb{P}(B)$ increases strictly, while $\mathbb{P}(A\backslash B)$  falls strictly. In this case, the risk‐averse agent begins to take on an increasing share of the risk, while the risk‐seeker's share correspondingly shrinks.
    \end{enumerate}
     The proof is non-trivial and the details are provided in Proposition \ref{proposition:example} of Appendix \ref{appendix:A}.
    For the special case $\alpha=2$ and $\beta=3$, Figure \ref{fig:example_2} illustrates the trend of $\mathbb{P}(B)$ and $\mathbb{P}(A\setminus B)$ as $\mathbb{P}(A)$ varies. 
    \begin{figure}[ht]
    \centering
    \includegraphics[width=8cm]{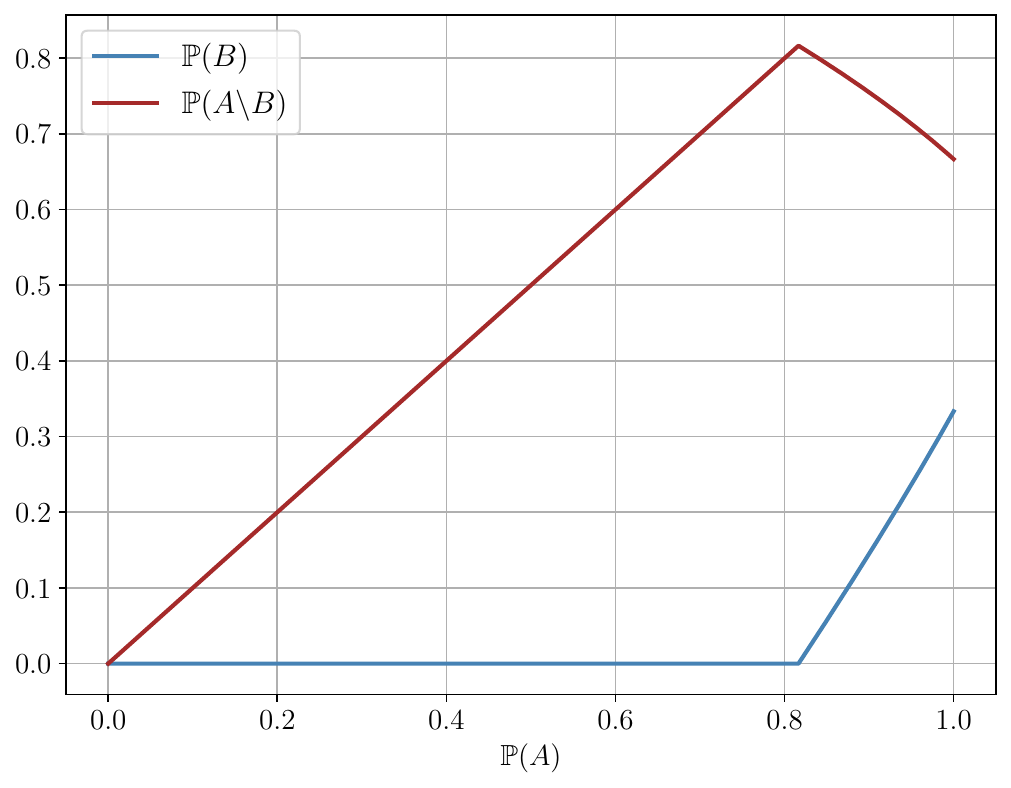}
    \caption{An illustration of Example \ref{example:6}.}
    \label{fig:example_2}
    \end{figure} 
\end{example}

 It is clear that examining an economy without aggregate uncertainty, where the total endowment is constant, is a special case of Theorem \ref{theorem:ele}; as shown in the following corollary.
\begin{corollary}
    Suppose that $\mathcal{X}=L^+$. For any  $a\in \mathbb{R}_{+}$, it holds that
    \begin{align}\label{eq:constant}
        \rho_{h_1}\dsquare \rho_{h_2}(a)=a h_1\dsquare h_2(1).
    \end{align}
    Moreover, a Pareto-optimal allocation is given by $X_1=\mathbb{1}_{B}$ and $X_2=\mathbb{1}_{A\setminus B}$ with $h_1(\mathbb{P}(B))+h_2(\mathbb{P}(A\setminus B))=h_1 \dsquare h_2(1)$.
\end{corollary}
The result directly follows from Theorem \ref{theorem:ele} by taking $A=\Omega$. If additionally $h_1\geq \tilde{h}_2$, implying that $h_1 \dsquare h_2(1)=1$, then $\rho_{h_1}\dsquare \rho_{h_2}(a)=a$ for any positive constant $a$. In this case, any constant split $(b, a-b)$ with $0\leq b\leq a$ is Pareto optimal. 

In fact, the equality \eqref{eq:constant} relies on the richness of the probability space to ensure that $\p(B)$ can achieve the infimum. 
On non-atomless space (e.g., finite probability space), the required probability level may be unattainable and the equality can fail; see more details in Example \ref{example:finite}.

\begin{example}[Counter-example in a finite probability space]\label{example:finite}
Define a probability space $(\Omega, \mathcal{F}, \mathbb{P})$, where $\Omega=\left\{\omega_1, \omega_2\right\}$, and $\mathbb{P}$ is such that $\mathbb{P}(\left\{\omega_1\right\})=1/6$, $\mathbb{P}(\left\{\omega_2\right\})=5/6$.
Suppose that two agents have distortion functions given by 
    \begin{align*}
        {h}_1 (x)=
        \begin{cases}
        2 x &x\in [0, 0.5]\\
         1 &x\in [0.5, 1];
        \end{cases}~~\  
        h_2(x)=
        \begin{cases}
        0 &x\in [0, 2/3]\\
        3x-2 &x\in [2/3, 1].
        \end{cases}
    \end{align*} 
    We can show that
    \begin{align*}
        h_1 \dsquare h_2(x)=
        \begin{cases}
        0 &x\in [0, 2/3]\\
         2x-4/3 &x\in [2/3, 1].
        \end{cases}
    \end{align*} 
    Assume that two agents are sharing a constant $1$. The only possible allocation $(X_1, X_2)$ of $1$ is in the form of  $(a\mathbb{1}_{\left\{\omega_1\right\}}+ b\mathbb{1}_{\left\{\omega_2\right\}}, (1-a)\mathbb{1}_{\left\{\omega_1\right\}}+(1-b)\mathbb{1}_{\left\{\omega_2\right\}}) $, where $0\leq a,b\leq 1$. If $a\leq b$, then 
    $\rho_{h_1}(X_1)+\rho_{h_2}(X_2)=a+ (b-a)h_2(5/6)+ (1-b)+(b-a)h_1(1/6)=1$. 
     If $a\geq b$, then 
     $\rho_{h_1}(X_1)+\rho_{h_2}(X_2)=b+(a-b)h_1(1/6)+(1-a)+(a-b)h_2(5/6)=1-1/6(a-b)\geq 5/6$.  Also, it is straightforward to verify that $\rho_{h_1\dsquare h_2}(1)=2/3$.
    This shows that $\rho_{h_1} \dsquare \rho_{h_2}(1)> \rho_{h_1\dsquare h_2}(1)$.
\end{example}

Having established the equality \eqref{eq:thm_indi} for indicator risks, we now turn to more complex distributions. When the total risk involves multiple components rather than a single Bernoulli variable, the exact equality no longer necessarily holds.
Even so, this broader setting yields useful insights, particularly regarding whether and when the inf-convolution can still attain the benchmark value. 
Proposition \ref{theorem:ineq} demonstrates that such attainability persists in a structured class of random risks. Specifically, when the total risk is composed of disjoint indicator components, there always exists a feasible allocation achieving the value of $\rho_{h_1\dsquare h_2}$.

\begin{proposition}\label{theorem:ineq}
    Suppose that $h_1, h_2\in \mathcal{H}$ and $\mathcal{X}=L^+$.
Let $X=a\mathbb{1}_A+b\mathbb{1}_B$, where
$A, B\in \mathcal{F}$ are disjoint and $a, b\in\mathbb{R}_{+}$ are constants. Then it holds that 
    \begin{align}\label{ineq:disjoint}
        \rho_{h_1}\dsquare \rho_{h_2}(X)\leq  \rho_{h_1\dsquare h_2}(X).
    \end{align}
\end{proposition}
\begin{proof}
   Without loss of generality, we assume that $a\leq b$. Then $X$ can be reformulated as $X=a \mathbb{1}_{A'}+(b-a)\mathbb{1}_{B'}$, where $A'=A\cup B$ and $B'=B$. clearly, $B' \subseteq A'$, which implies that $(a\mathbb{1}_{A'}, (b-a)\mathbb{1}_{B'})$ is comonotonic.
   Let $p_{1}=\p(A')$ and $p_{2}=\p(B')$. By comonotonic additivity of $\rho_{h_1\dsquare h_2}$, we have 
   $$
   \rho_{h_1\dsquare h_2}(X)=a {h_1\dsquare h_2}(p_1)+(b-a){h_1\dsquare h_2}(p_2).
   $$
   Define the set 
   $$
   S_i=\left\{t \in [0,p_i]: h_1(t)+h_2(p_i-t)=h_1 \dsquare h_2(p_i)\right\}.
   $$
   Let $u_i\in S_i$ for each $i\in [2]$. 
   Our aim is to construct an allocation $(Y,Z)$ of $X$ such that $\rho_{h_1}(Y)+\rho_{h_2}(Z)=a h_1 \dsquare h_2(p_1)+ (b-a)h_1 \dsquare h_2(p_2)$.
    We construct an allocation of $X$ as follows:
    \begin{align*}
        Y=a\mathbb{1}_C+b\mathbb{1}_D \  \text{ and } \  Z=a\mathbb{1}_{A'\backslash C}+b\mathbb{1}_{B'\backslash D} \ \text{ with } \ \p(C)=u_1 \ \text{ and } \ \p(D)=u_2.
    \end{align*}
    In fact, the construction of $C$ and
    $D$ vary with the magnitude relationship between $u_1$ and $u_2$, as well as $p_1-u_1$ and $p_2-u_2$. 
    Next, we will show the details about how to construct $C$ and
    $D$ in different cases. We consider the following three cases.

    Case 1: $u_1\leq u_2$ and $p_1-u_1\leq p_2-u_2$. This case cannot happen unless $\p(A)=0$, then it reduces to Theorem \ref{theorem:ele}. Therefore, \eqref{ineq:disjoint} holds trivially.
    
    Case 2: $u_1\leq u_2$ and $p_1-u_1\geq p_2-u_2$. It implies that  $u_1\leq u_2$ since $p_1\geq p_2$.
    Take $C\subseteq B$ with $\p(C)=u_1$ and $D=C\cup E$ with $E \subseteq B'\backslash C$ and $\p(E)=u_2-u_1$. Hence, $\p(D)=u_2$. Also, we have $B'\backslash D \subseteq A'\backslash C $.

    Case 3: $u_1\geq u_2$ and $p_1-u_1\geq p_2-u_2$. Thus, $u_1\geq u_2 \geq p_2-p_1+u_1$. 
    Let $D\subseteq B'$ with $\p(D)=u_2$. The take $C=D \cup E$ and $E\subseteq A'\backslash B'$ with $\p(E)=u_1-u_2$. Hence, $\p(C)=u_1$. 
    In this case, we have $B'\backslash D \subseteq A'\backslash C$.

    Case 4: $u_1\geq u_2$ and $p_1-u_1\leq p_2-u_2$. It implies that $u_1 \geq p_1-p_2+u_2$. 
    Let $D\subseteq B'$ with $\p(D)=u_2$. 
    Take $C= D\cup (A'\backslash B')\cup E$ and $E\subseteq B'\backslash D$ with $\p(E)=u_1-u_2-(p_1-p_2)$. Thus, we have $\p(C)=u_1$ and $A'\backslash C\subseteq B'\backslash D$.

    By above constructions, both $(a\mathbb{1}_C, (b-a)\mathbb{1}_D)$ and $(a\mathbb{1}_{A'\backslash C}, (b-a)\mathbb{1}_{B'\backslash D})$ are comonotonic since the underlying events are nested. Then it follows that 
    \begin{align*}
        \rho_{h_1}(Y)+\rho_{h_2}(Z)&=a \big(h_1(u_1)+h_1(u_2)\big)+(b-a)\big(h_2(p_1-u_1)+h_2(p_2-u_2)\big)\\&=
    a h_1 \dsquare h_2(p_1)+ (b-a)h_1 \dsquare h_2(p_2)=\rho_{h_1 \dsquare h_2}(X).
    \end{align*}
    Consequently, $\rho_{h_1}\dsquare \rho_{h_2}(X)\leq \rho_{h_1 \dsquare h_2}(X)$.
\end{proof}
By translation invariance, one typically has $\rho_{h_1}\dsquare \rho_{h_2}(\mathbb{1}_A+c)\leq 
    \rho_{h_1\dsquare h_2}(\mathbb{1}_A)+c$ for any constant $c$, where the bound comes from handling the constant and the indicator separately. 
    Using Theorem~\ref{theorem:ineq}, a tighter bound can be obtained by treating the constant as an additional layer over a disjoint set; see Corollary \ref{coro:2}.
\begin{corollary}\label{coro:2}
    Suppose that $h_1, h_2\in \mathcal{H}$ and $c\in \mathbb{R}_{+}$. For any $A\in \mathcal{F}$, it holds that
    $$\max\left\{h_1 \dsquare h_2(\mathbb{P}(A)), c\, h_1 \dsquare h_2(1)\right\}\leq \rho_{h_1}\dsquare \rho_{h_2}(\mathbb{1}_A+c)\leq 
    \rho_{h_1\dsquare h_2}(\mathbb{1}_A+c).$$
\end{corollary}
\begin{proof}
    The first inequality directly follows from Theorem \ref{theorem:ele}. 
    It is immediate to show the second inequality by applying Theorem \ref{theorem:ineq} and replacing $B=A^c$ and $a=1+c=1+b$.
\end{proof}

The following proposition provides some results for random variables $X$ satisfying $\p(X>0)\leq a$ for some $a\in[0,1]$. We first introduce some notations for convenient. For any $h_1, h_2\in \mathcal{H}$, let 
$x_c=\sup\left\{x\in [0,1]: h_1 \dsquare h_2(x)=h_2(x)\right\}$ and $x_d= \sup \left\{x>0:{h_2}_{+}'(x)\leq 1\right\}$.

\begin{proposition}\label{proposition:half}
    Suppose that $\mathcal{X}=L^+$ and $h_1, h_2\in \mathcal{H}$.  
    For $X\in\mathcal{X}$ with $\p(X>0)\leq \alpha$, $\alpha\in [0,1)$,
    the following hold.
    \begin{enumerate}
        \item[(i)] If $h_1 \dsquare h_2$ is convex and $\alpha\leq 1/2$, then $\rho_{h_1}\dsquare \rho_{h_2}(X)\geq \rho_{h_1\dsquare h_2}(X)$.
        \item[(ii)] If 
        $h_2$ is convex and $\alpha\leq x_c/2$, then $\rho_{h_1}\dsquare \rho_{h_2}(X)= \rho_{h_2}(X)$.
        \item[(iii)] If $h_1$ is concave, $h_2$ is convex and $\alpha \leq\max\left\{x_c/2, x_d\right\}$,
        then $\rho_{h_1}\dsquare \rho_{h_2}(X)= \rho_{h_2}(X)$.
    \end{enumerate}
\end{proposition}
\begin{proof}
(i) For any allocation $(X_1, X_2)\in \mathbb{A}_2(X)$ and $a \in (0, \infty)$, we have $\mathbb{E}(X_2 \wedge a)+ \mathbb{E}(X_1 \wedge a)\geq \mathbb{E}(X\wedge a)$, implying that $$ \int_{0}^{a} (\mathbb{P}(X_1>t)+\mathbb{P}(X_2>t))\mathrm{d}t\geq \int_{0}^{a} \mathbb{P}(X>t)\mathrm{d}t.$$ 
     By integral majorization theorem \cite[Theorem 12.15]{peajcariaac1992convex} and convexity of $h_1\dsquare h_2$, it holds that
\begin{align}\label{ineq:13}
         \int_{0}^{a} h_1\dsquare h_2(\mathbb{P}(X_1>t)+\mathbb{P}(X_2>t))\mathrm{d}t \geq  \int_{0}^{a} h_1\dsquare h_2(\mathbb{P}(X>t))\mathrm{d}t.
     \end{align}
     Note that $\p(X>0)\leq 1/2$ ensures that $h_1\dsquare h_2$ is well-defined.
     Then it follows that 
\begin{align*}
    \rho_{h_1}(X_1)+\rho_{h_2}(X_2)
    &=\int_0^\infty h_1(\p(X_1>t))+h_2(\p(X_2>t))
    \mathrm{d}t\\
    &\geq \int_0^\infty h_1\dsquare h_2(\p(X_1>t)+\p(X_2>t))
    \mathrm{d}t
    \geq \rho_{h_1\dsquare h_2}(X),
\end{align*}
which implies that $\rho_{h_1}\dsquare \rho_{h_2}(X)\geq \rho_{h_1\dsquare h_2}(X)$.

(ii)  Clearly, $\rho_{h_1}\dsquare \rho_{h_2}(X)\leq \rho_{h_2}(X)$. It suffices to show the converse direction.
We first show that for any $x\leq x_0$, it holds that $h_1\square h_2(x)=h_2(x)$. By definition of $x_0$, it follows that $f_{x_0}(y):=h_1(y)+h_2(x_0-y)-h_2(x_0)\geq 0$ for $y \leq x_0$. Then for $x\leq x_0$, we have
    \begin{align}\label{eq:thm4}
        h_1(y)+h_2(x-y)-h_2(x)=f_{x_0}(y)
         -h_2(x_0-y)+h_2(x-y)+h_2(x_0)-h_2(x) \geq 0.
    \end{align}
    The inequality \eqref{eq:thm4} holds due to  convexity of $h_2$, implying that $h_2(x_0)-h_2(x)\geq h_2(x_0-y)-h_2(x-y)$.
    For any allocation $(X_1, X_2)\in \mathbb{A}_2(X)$, it holds that 
    \begin{align*}
    \rho_{h_1}(X_1)+\rho_{h_2}(X_2)
    &\geq \int_0^\infty h_1\dsquare h_2(\p(X_1>t)+\p(X_2>t))
    \mathrm{d}t\\
    &=\int_0^\infty h_2(\p(X_1>t)+\p(X_2>t))
    \mathrm{d}t\\
    &\geq \int_0^\infty h_2(\p(X>t))
    \mathrm{d}t=\rho_{h_2}(X).
    \end{align*}
    The equality holds due to $\p(X>0)\le x_0/2$, which implies that 
     $h_1\dsquare h_2 (g_1(t)+g_2(t))= h_2(g_1(t)+g_2(t))$ for $t\geq 0$.
    The second inequality follows from \eqref{ineq:13} and convexity of $h_2$.
     Thus, it follows that
     $\rho_{h_1}\dsquare \rho_{h_2}(X)\geq \rho_{h_2}(X)$ for 
     any $X \in \mathcal{X}$ with $\p(X>0)\le x_0/2$. Therefore, the desired result is obtained.

     (iii) Following from (ii), it suffices to show that the result holds for $\alpha=x_d$.
     For any allocation $(X_1, X_2)$ of $X\in \mathcal{X}$, it holds that 
     \begin{align}\notag
     &\int_{0}^{\infty}h_1 (\mathbb{P}(X_1>t)) \mathrm{d}t + \int_{0}^{\infty}h_2 (\mathbb{P}(X_2>t)) \mathrm{d}t \\ \label{ineq:conca_conve}
     &\geq \int_{0}^{\infty} \left(
     \mathbb{P}(X_1>t) +h_2(\mathbb{P}(X>t))+ {h_2}_{+}'(\mathbb{P}(X>t))(\mathbb{P}(X_2>t)-\mathbb{P}(X>t))\right) \mathrm{d}t.
     \end{align}
     The inequality holds due to concavity of $h_1$ and convexity of $h_2$, implying that $h_1(x)\geq x$ and  
     $h_2(y)\geq h_2(x)+ {h_2}_{+}'(x)(y-x)$ for any $x,y\in [0,1]$; see Theorem 25.1 of \cite{rockafellar1997convex}.
     Next we show it always holds that ${h_2}_{+}'(\mathbb{P}(X>t))\leq 1$. 
     If $x_d=0$, then $h_2$ is the identity function and  ${h_2}_{+}'(x)=1$. If $x_d>0$, then ${h_2}_{+}'(\mathbb{P}(X>t))\leq 1$ since $\mathbb{P}(X>0)\leq x_d$.
     By the fact of $\mathbb{P}(X_2>t)\leq \mathbb{P}(X>t)$ for $t> 0$ and \eqref{ineq:conca_conve}, we have
     \begin{align*}
         \rho_{h_1}(X_1)+\rho_{h_2}(X_2)&\geq \int_{0}^{\infty} \left(
      h_2(\mathbb{P}(X>t))+ \mathbb{P}(X_1>t)+\mathbb{P}(X_2>t)-\mathbb{P}(X>t)\right) \mathrm{d}t \\&= \int_{0}^{\infty}h_2(\mathbb{P}(X>t)) \mathrm{d}t= \rho_{h_2}(X).
     \end{align*}
     The above result implies that $\rho_{h_1}\dsquare \rho_{h_2}(X)= \rho_{h_2}(X)$ with $\p(X>0)\leq x_d$. Consequently, the desired result is obtained.
\end{proof}

Proposition \ref{proposition:half} shows that 
when the loss probability is sufficiently small, the efficient arrangement assigns the entire risky slice to the risk-seeking agent. Intuitively, small-probability losses are nearly ``free'' under a convex  distortion function, so letting the risk-seeker absorb them minimizes the total risk value. Moreover, if the condition $h_1\geq \tilde {h_2}$ holds, the relative strength of risk aversion over risk seeking is strong enough that the optimal arrangement assigns the entire risk (not just the small layer) to the risk-seeking agent;  see Theorem \ref{theorem:1}.

\section{Applications to the original problem}\label{Sec:6}
In this section, we return to the $n$-agent risk-sharing problem and demonstrate how the results established in Sections~\ref{Sec:4} and~\ref{Sec:5} can be applied to characterize optimal allocations in mixed economies, including both risk-averse and risk-seeking agents.
Recall that $g_1=\bigwedge_{i=1}^{m}h_i$ and $g_2=\square_{j=m+1}^{n}h_j$ for $m\leq n$.

\begin{proposition}\label{proposition:n_averse}
    Let $\mathcal{X}=L^+$. Assume that $h_i\in \mathcal{H}$ are continuous for all $i\in [n]$, with $h_i$ concave for $i\in[m]$ and $h_j$ convex for $j\in [n]\setminus[m]$, where $m\leq n$.
     If $g_1\geq \tilde{h}_{i}$ for some $i\in [n]\setminus[m]$,
    then for any $X\in\mathcal{X}^\perp$, we have 
    \begin{align*}
    \dsquare_{i=1}^{n}\rho_{h_i}(X)=\rho_{g_2}(X).
    \end{align*}
\end{proposition}
\begin{proof}
    We first note that the inf-convolution $\dsquare_{i=1}^{n}\rho_{h_i}(X)$ is associative; see Lemma 2 of \cite{liu2020inf}.
    Without loss of generality, we assume $g_1\geq \tilde{h}_{m+1}$. 
    By Theorem \ref{theorem:no_concave_convex}, we have $\rho_{g_1}\dsquare \rho_{h_{m+1}}(X)=\rho_{h_{m+1}}(X)$. Thus,
    \begin{align*}
        \dsquare_{i=1}^{n}\rho_{h_i}(X)=\rho_{g_1}\dsquare \rho_{h_{m+1}}\dsquare \cdots \dsquare \rho_{h_{n}}(X)=\rho_{g_2}(X).
    \end{align*}
    Therefore, the desired result is obtained.
\end{proof}
Proposition \ref{proposition:n_averse} states that
if the least risk-averse agent in the cautious group is more conservative than at least one risk-seeking agent is adventurous (formally, if $g_1\geq \tilde{h}_{i}$
 for some convex $h_i$), then the efficient split assigns all the randomness to the risk-seeking side.
Intuitively, once there exists a single risk-seeker willing to absorb the entire uncertain part when they are sharing the risk with the most tolerant risk-averse agent, the problem effectively turns into a ``betting game'' among the risk-seeking agents.
In this case, the optimal allocation is counter-monotonic, that is, risk-averse agents bear nothing and
risk-seekers bet on who takes the total risk. 

When such dominance condition fails and the risk seeking dominates risk aversion, the optimal risk sharing would change accordingly. 
As we know, a full analysis with arbitrary distortion functions is challenging, so we focus on a tractable subclass of piecewise linear distortions. 
Building on the two-agent analysis in Section \ref{Sec:5}, the following proposition provides explicit formulas for the $n$-agent inf-convolution across different settings.

\begin{proposition}\label{proposition:n_linear}
     Suppose that $\mathcal{X}=L^+$ and $m,n\in \mathbb{N}$. Let $h_i(x)=\min \left\{ x/b_i, 1\right\}$ for $i\in[m]$ and 
     $h_i(x)=\max \left\{0, (x-a_i)/(1-a_i)\right\}$ for $i\in [n]\setminus[m]$ over $x\in [0,1]$, where $a_i,b_i\in(0,1)$ with $\sum_{i=m+1}^n a_i\leq 1$. Denote by $b=\bigvee_{i=1}^m b_i$.
     Then the following hold for $X\in \mathcal{X}^\perp$.
     \begin{itemize}
         \item[(i)] If $a_i+b\leq 1$ for some $i\in[n]\backslash[m]$, then 
         \begin{align*}
             \dsquare_{i=1}^{n}\rho_{h_i}(X)=\rho_{g}(X), ~\text{where}~g(t)=\max\left\{0, \frac{t-\sum_{i=m+1}^n a_i}{1-\bigvee_{i=m+1}^n a_i}\right\}, ~t\in[0,1].
         \end{align*}
         \item[(ii)] If $\bigwedge_{i=m+1}^n a_i+b>1$, then 
         \begin{align*}
             \dsquare_{i=1}^{n}\rho_{h_i}(X)=\rho_{g}(X), ~\text{where}~g(t)=\max\left\{0, \frac{t-\sum_{i=m+1}^n a_i}{b}\right\}, ~t\in[0,1].
         \end{align*}
     \end{itemize}
\end{proposition}
\begin{proof}
    Let $\ell(t)=\max\left\{0, (t-\sum_{i=m+1}^n a_i)/({1-\bigvee_{i=m+1}^n a_i})\right\}$.
    It is straightforward to verity that $\dsquare_{i=m+1}^{n}h_i(x)=\ell(x)$.
    
    (i) The result directly follows from Proposition \ref{proposition:n_averse}. 

    (ii) To make presentations more precise, we first introduce some notations, define
    \begin{align*}
        \ell_1(x)=\max\left\{0, \frac{x-\sum_{i=m+1}^n a_i}{1-\sum_{i=m+1}^n a_i}\right\}, ~\text{and}~ \ell_2(x)=\min\left\{\frac{x}{1-\bigvee_{i=m+1}^n a_i}, 1\right\}
        ~t\in[0,1].
    \end{align*}
    Then it can be verified that $\ell_1\dsquare \ell_2(x)=\ell(x)$ for $x\in[0,1]$.
    Note that $\bigwedge_{i=m+1}^n a_i+b>1$ implies that $\sum_{i=m+1}^n a_i+b>1$. 
    By Proposition~\ref{theorem:pie_linear}, we have 
    \begin{align}\label{eq:app_2}
        \rho_{\ell_1}\dsquare \rho_{\ell_2}(X)=\rho_{\ell}(X).
    \end{align}
    Let $h=\bigwedge_{i=1}^m h_i$.
    It follows that 
    \begin{align*}
        \dsquare_{i=1}^{n}\rho_{h_i}(X)&=\rho_{h}\dsquare \left(\dsquare_{i=m+1}^{n}\rho_{h_i}\right)(X)
        =\rho_{h}\dsquare \rho_{\ell}(X)=\rho_{h}\dsquare\left(\rho_{\ell_1}\dsquare \rho_{\ell_2}\right)(X)\\&=\rho_{h} \dsquare \rho_{\ell_1}(X)=
        \rho_{\ell}(X),
    \end{align*}
    where $\ell(t)= h\dsquare \ell_1(t)=
    \max\left\{0, (t-\sum_{i=m+1}^n a_i)/b\right\}$.
    The first equality follows from Theorem \ref{theorem:1}.
    By applying \eqref{eq:app_2} and Proposition~\ref{theorem:pie_linear}, we  obtain the last three equalities.
\end{proof}

As shown in Proposition \ref{proposition:n_linear}, when risk aversion dominates sufficiently, the risk-seeking group absorbs all randomness, consistent with Proposition~\ref{proposition:n_averse}.
In contrast, when risk seeking dominates risk aversion, as described by condition (ii), the pattern of optimal risk sharing changes substantially. The risk-averse participants continue to share risk comonotonically within their own group, but they would collectively bet with the risk-seeking group. The resulting allocation is counter-monotonic across groups, reflecting a betting structure in which the risk-seeking side and the risk-averse side effectively compete over who absorbs the uncertainty.

\section{Conclusion}\label{Sec:7}

In this paper, we study risk sharing in economies where agents need not be all risk-averse or all risk-seeking. 
This mixed setting is economically important with capturing the coexistence of cautious and speculative participants, and mathematically challenging because the underlying distortion risk measures are neither all convex nor all concave.

With mixed attitudes in place (some are risk-averse and others are risk-seeking), we establish a two-agent reduction by aggregating risk attitudes on each side into representative distortions; see Theorem \ref{proposition_1}. The reduced problem can be solved via the inf-convolution $\rho_{h_1}\dsquare \rho_{h_2}$.
One of our main results (Theorem \ref{theorem:1}) provides   necessary and sufficient conditions for existence of optimal allocations. 
When the conditions  hold, the concave-convex case admits explicit solutions, in which the risk-seeking side bears all the risk, as shown in Theorem \ref{theorem:no_concave_convex}.

We also consider the case of  
restricted allocations that are nonnegative.  This setting ensure a well-defined risk sharing problem and  corresponds to the natural assumption of no profit from pure losses.
We   compare the constrained value $\rho_{h_1}\dsquare \rho_{h_2}(X)$
 with the benchmark $\rho_{h_1\dsquare h_2}(X)$.
 Our analysis yields constructive bounds and equality conditions for specific distortion families (Theorems \ref{theorem:coin} and Proposition \ref{theorem:pie_linear}) and for canonical risks (Theorems \ref{theorem:ele} and Proposition \ref{theorem:ineq}). Finally, we apply these results to resolve the original multi-agent problem.

\section*{Acknowledgements}

Mario Ghossoub acknowledges financial support from the Natural Sciences and Engineering Research Council of Canada (RGPIN-2024-03744). Ruodu Wang acknowledges financial support from the Natural Sciences and Engineering Research Council of Canada (RGPIN-2024-03728, CRC-2022-00141).


\bibliographystyle{apalike}
\bibliography{biblio}
 
\appendix

\section{Details in Example \ref{example:6}}\label{appendix:A}

The details in Example \ref{example:6} are guaranteed by the following proposition.
\begin{proposition}\label{proposition:example}
    Suppose that $h_1(x)=1-(1-x)^{\alpha}$ and $h_2(x)=x^\beta$ for $x\in [0,1]$
    with $1<\alpha < \beta$. Let $p_0=(\alpha/\beta)^{\frac{1}{\beta-1}}$. Then there exists $y^*(x)$ such that $h_1(y^*(x))+ h_2(x-y^*(x))= h_1\dsquare h_2(x)$. 
    Moreover, 
    \begin{enumerate}
        \item[(i)] If $x\leq p_0$, $y^*(x)=0$ and $x-y^*(x)=x$.
        \item[(ii)] If $x> p_0$, $y^*(x)$ satisfies that $ \alpha (1-y^*(x))^{\alpha-1}= \beta (x-y^*(x))^{\beta-1}$. Moreover, $y^*(x)$ is increasing and $x-y^*(x)$ is decreasing.
    \end{enumerate}
\end{proposition}
\begin{proof}
    Let $g(y)=h_1(y)+h_2(x-y)$ for $x\in [0,1]$ and $y\in [0,x]$. Then 
    \begin{align*}
        g'(y)=h_1'(y)-h_2'(x-y)=\alpha (1-y)^{\alpha-1}-\beta (x-y)^{\beta-1}.
    \end{align*}
    Let $R(y)=\frac{\alpha (1-y)^{\alpha-1}}{\beta (x-y)^{\beta-1}}$. By straightforward calculation, we have 
    \begin{align*}
        \ln R(y)= \ln \alpha -\ln \beta + (\alpha-1)\ln(1-y)-(\beta-1)\ln (x-y) \ \text{and}\ 
        (\ln R(y))'= \frac{\beta-1}{x-y}-\frac{\alpha-1}{1-y}.
    \end{align*}
    Since $\alpha<\beta$, $(\ln R(y))'>0$, thus $ \ln R(y)$ is increasing. Also, $\ln R(0)=\ln \alpha -\ln \beta -(\beta-1)\ln x$.
    Next, we consider two cases:
    \begin{enumerate}
        \item[(i)] If $x\leq p_0$, then $\ln R(y)\geq \ln R(0)\geq 0$ for $y\in [0,x]$. This implies that $R(y)\geq 1$ and $g'(y)\geq 0$, thus $g(y)$ is increasing. Hence, the infimum of $g(y)$ over $y\in [0,x]$ attains at $y^*=0$. 
        \item[(ii)] If $x> p_0$, then $\ln R(0)<0$ and $\ln R(y)\rightarrow +\infty$ as $y\downarrow x$. Thus, $g'(y)$ would be first less than zero and then cross zero. Then first order condition gives:
        \begin{align}\label{eq:first}
        g'(y^*)=h_1'(y^*)-h_2'(x-y^*)=\alpha (1-y^*)^{\alpha-1}-\beta (x-y^*)^{\beta-1}=0.
    \end{align}
    Write $F(x, y^*)=\alpha (1-y^*)^{\alpha-1}-\beta (x-y^*)^{\beta-1}$. Then it follows that 
    \begin{align*}
        \frac{\mathrm{d} y^*}{\mathrm{d} x}&=-\frac{\partial F / \partial x}{\partial F\ \partial y}=\frac{\beta (\beta-1)(x-y^*)^{\beta-2}}{\beta (\beta-1)(x-y^*)^{\beta-2}-\alpha (\alpha-1)(1-y^*)^{\alpha-2}}
        =\frac{1}{1-\frac{\alpha-1}{\beta-1}\frac{x-y^*}{1-y^*}}> 0.
    \end{align*}
    The last equality can be obtained from
    \eqref{eq:first}.
    By straightforward calculation,
    we have 
    \begin{align*}
        \frac{\mathrm{d}(x-y^*)}{\mathrm{d} x} = 1- \frac{\mathrm{d} y^*}{\mathrm{d} x}=\frac{-(\alpha-1)(x-y^*)}{(\beta-1)(1-y^*)-(\alpha-1)(x-y^*)}< 0.
    \end{align*}
    \end{enumerate}
    Consequently, the desired result is obtained.
\end{proof}

\section{Basic properties of the constrained inf-convolution}
Let $(\Omega, \mathcal{F}, \mathbb{P})$ be an atomless probability space and $\mathcal{X}=L^+$ be the set
of nonnegative random variables in this space. 
For risk functionals $\rho_1,\rho_2:\mathcal{X}\mapsto [0, \infty)$, define the constrained inf-convolution 
$$
\rho_1 \dsquare \rho_2(X)=\inf \left\{\rho_1(Y)+ \rho_2(X-Y), 0\leq Y\leq X\right\}, ~~ X\in \mathcal{X}.
$$
The constraint $0\leq Y\leq X$ ensures $Y, X-Y\in \mathcal{X}$.
Here we record several properties that can pass from $\rho_1$ and $\rho_2$ to their constrained inf-convolution
$\rho_1 \dsquare \rho_2$. 
Our focus in this appendix is the constrained formulation. The corresponding properties in the unconstrained case have been well studied in \cite{liu2020inf}. 

\begin{lemma}\label{lemma:2_a}
   Suppose that $\mathcal{X}=L^+$ and $X\in \mathcal{X}$.
   Let $\rho_1, \rho_2$ be two risk functionals. 
   \begin{enumerate}
       \item[(i)] If $\rho_1$ and $\rho_2$ are monotone, then $\rho_1 \square \rho_2$ is monotone.
       \item[(ii)]If $\rho_1$ and $\rho_2$ are uniformly continuous, then $\rho_1 \square \rho_2$ is uniformly continuous.
       \item[(iii)]If $\rho_1$ and $\rho_2$ are monotone and one of them is continuous from above, then $\rho_1 \square \rho_2$ is continuous from above.
   \end{enumerate}
\end{lemma}
\begin{proof}
\begin{enumerate}
    \item[(i)] 
    Suppose that $X, Y\in \mathcal{X}$ with $X\leq Y$. For any $\varepsilon>0$, there exists $Z_{Y}\in \mathcal{X}$ such that $0\leq Z_{Y}\leq Y$ and 
    $$
    \rho_1(Z_{Y})+\rho_2(Y-Z_{Y})\leq \rho_1 \dsquare \rho_2(Y)+\varepsilon.
    $$
    Let $\tilde{Z}=Z_Y \wedge X$ (so $0\leq \tilde{Z}\leq X$). Thus, $\tilde{Z}\leq Z_Y$ and $X-\tilde{Z}\leq Y-Z_Y$. 
    Then
    $$
    \rho_1 \dsquare \rho_2(X)\leq \rho_1(\tilde{Z})+\rho_2(X-\tilde{Z})\leq \rho_1(Z_{Y})+\rho_2(Y-Z_{Y})\leq \rho_1 \dsquare \rho_2(Y)+\varepsilon.
    $$
    Since $\varepsilon>0$ is arbitrary, we obtain $\rho_1 \dsquare \rho_2(X)\leq \rho_1 \dsquare \rho_2(Y)$.
    \item[(ii)]
    Since $\rho_1$ and $\rho_2$ are uniformly continuous, for any $\varepsilon>0$, there exists $\delta_i$ for $i=1,2$ such that for all $X, Y\in \mathcal{X}$, $\|X-Y\| \leqslant \delta_i$ implies $\left|\rho_i(X)-\rho_i(Y)\right| \leqslant \varepsilon/3$. 
    By definition of $\rho_1\dsquare \rho_2$, there exists $Z_X\in [0,X]$ such that $\rho_1(Z_X)+\rho_2(X-Z_X)\leq \rho_1\dsquare \rho_2(X)+\varepsilon/3$. 
    Let $Z'=Z_X\wedge Y$ and
    $\delta=\min \left\{\delta_1, \delta_2/2\right\}$. 
    For any $X, Y \in \mathcal{X}$ with $\|X-Y\| \leqslant \delta_1$, we have 
    $$\|Z'-Z_X\| \leq \|X-Y\|\leq \delta ~\text{and}~ \|(X-Z_X)-(Y-Z')\|\leq \|X-Y\|+\|Z'-Z_X\| \leq \delta_2.$$ 
    Then it follows that
    \begin{align*}
        &|\rho_1 \dsquare \rho_2(X)-\rho_1 \dsquare \rho_2(Y)| \\&\leq |\rho_1(Z_X)+\rho_2(X-Z_X)-\rho_1(Z')-\rho_2(Y-Z')|+\varepsilon/3 \\
        &\leq |\rho_1(Z_X)-\rho_1(Z')|+ |\rho_2(X-Z_X)-\rho_2(Y-Z')|+\varepsilon/3 \leq \varepsilon.
    \end{align*}
    Hence, $\rho_1\dsquare \rho_2$ is uniformly continuous.
    
    \item[(iii)]Without loss of generality, we assume $\rho_2$ is continuous from above. 
    By part (i), we have $\rho_1 \square \rho_2$ is monotone. Let $\left\{X_n\right\}_{n \in \mathbb{N}}$ be a sequence in $\mathcal{X}$ such that $X_n \downarrow X$ as $n \rightarrow \infty$. 
    For any $\varepsilon>0$ and choose $Z\in \mathcal{X}$ with $0 \leq Z \leq X$ such that
$$\rho_1\left(Z\right)+\rho_2\left(X-Z\right) \leq \rho_1 \dsquare \rho_2(X)+\varepsilon.
    $$

Since $X_n \downarrow X$ and $Z \leq X$, we have $Z\leq X_n$ for all $n$. Hence, 
$$
\limsup _{n\rightarrow\infty} \rho_1 \square \rho_2\left(X_n\right) \leq \rho_1\left(Z\right)+\limsup _{n\rightarrow\infty} \rho_2\left(X_n-Z\right)=\rho_1\left(Z\right)+\rho_2\left(X-Z\right) .
$$
Thus, $\limsup _{n\rightarrow\infty} \rho_1 \square \rho_2\left(X_n\right) \leq \rho_1 \square \rho_2\left(X\right)$.
On the other hand, since $\rho_1 \square \rho_2$ is monotone, we have $\rho_1 \square \rho_2\left(X_n\right) \geqslant \rho_1 \square \rho_2(X)$. This implies
$$
\lim _{n \rightarrow \infty} \rho_1 \square \rho_2\left(X_n\right)=\rho_1 \square \rho_2(X).
$$
Hence, $\rho_1 \square \rho_2$ is continuous from above.
\end{enumerate}
The desired result is obtained.
\end{proof}

Lemma \ref{lemma:2_a} shows that monotonicity of $\rho_1\dsquare \rho_2$ under nonnegative allocations requires both $\rho_1$ and $\rho_2$ to be monotone. For general allocation sets, one monotone component is sufficient, as shown in Lemma 1 of \cite{liu2020inf}. The example below demonstrates that this sufficiency fails when allocations are constrained to be nonnegative.
\begin{example}
Take $(\Omega, \mathcal{F}, \mathbb{P})=([0,1], \mathcal{B}([0,1]), \lambda)$ where $\lambda$ is the Lebesgue measure.
Let $B\subset[0,1]$ satisfy $\lambda(B)=\frac{1}{2}$.
Define two risk functionals $\rho_1(X)=1-\mathbb{1}_{\left\{X \sim \operatorname{Bernoulli}\left(\frac{1}{2}\right)\right\}}$ and $\rho_2(X)=\mathbb{E}[X]$.
 Set
$$
X=c \mathbb{1}_B ~\text { with } c \in(0,1) ~ \text { and }~  Y=\mathbb{1}_B .
$$
Clearly $0 \leq X \leq Y$.
Any feasible $0\leq Z\leq X$ satisfies $Z=0$ on $B^c$ and $0 \leq Z \leq c<1$ on $B$.
Hence $\rho_1(Z)=1$ for all $0\leq Z\leq X$.
Therefore, $$\rho_1\dsquare \rho_2(X)=\rho_1(X)+\rho_2(0)=1>\rho_1(Y)+\rho_2(0)=0\geq \rho_1\dsquare \rho_2(Y).$$
In this case, $\rho_1\dsquare \rho_2$ is not monotone.
\end{example}

\end{document}